\documentclass[journal, letterpaper]{IEEEtran}
\usepackage{amsmath,amsfonts,amssymb,euscript, graphicx, 
epsfig,
enumerate,float,afterpage, subfigure, ifthen}%

\usepackage{url}
\usepackage{hyperref}

\newtheorem{thm}{Theorem}

\newtheorem{lem}{Lemma}

\newcommand{\expect}[1]{\mathbb{E}\left[#1\right]}
\newcommand{\defequiv}{\mbox{\raisebox{-.3ex}{$\overset{\vartriangle}{=}$}}}

\newcommand{\norm}[1]{||{#1}||}

\newcommand{\bv}[1]{{\boldsymbol{#1} }}
\newcommand{\script}[1]{{{\cal{#1} }}}

\begin{document}

\title
  {A Lyapunov Optimization Approach to Repeated Stochastic Games}
\author{Michael J. Neely\\University of Southern California\\\url{http://www-bcf.usc.edu/~mjneely}
\thanks{This paper was presented in part at the Allerton conference on communication, control, and computing, Monticello, IL, Oct. 2013 \cite{repeated-games-allerton}.} 
\thanks{The author is with the  Electrical Engineering department at the University
of Southern California, Los Angeles, CA.} 
\thanks{This work is supported in part  by one or more of:  the NSF Career grant CCF-0747525,  NSF grant 1049541,  the 
Network Science Collaborative Technology Alliance sponsored
by the U.S. Army Research Laboratory W911NF-09-2-0053.}
}

\markboth{}{Neely}

\maketitle

\begin{abstract}   
This paper considers a time-varying game with $N$ players.  Every time slot, players observe their own random events and then take a  control action.  The events and control actions affect the individual utilities earned by each player.   The goal is to maximize a concave function of time average utilities subject to equilibrium constraints.  Specifically, participating players are provided access to a common source of randomness from which they can optimally correlate their decisions.  The equilibrium constraints incentivize participation by ensuring that players cannot earn more utility if they choose not to participate.  This form of equilibrium is similar to the notions of Nash equilibrium and correlated equilibrium, but is simpler to attain.  A Lyapunov method is developed that solves the problem in an online \emph{max-weight} fashion by selecting actions based on a set of time-varying weights.  The algorithm does not require knowledge of the event probabilities and has polynomial convergence time.  A similar method can be used to compute a standard correlated equilibrium, albeit with increased complexity.  
\end{abstract} 

\section{Introduction}

Consider a repeated game with $N$ \emph{players} and one \emph{game manager}.  The game is played over an infinite sequence of time slots
 $t \in \{0, 1,2, \ldots\}$.  Every slot $t$ there is a \emph{random event vector} $\bv{\omega}(t) = (\omega_0(t), \omega_1(t), \ldots, \omega_N(t))$. 
 The game manager observes the full vector $\bv{\omega}(t)$, while each player $i \in \{1, \ldots, N\}$ observes only the component $\omega_i(t)$. 
 The value $\omega_0(t)$ represents information known only to the manager.  After the slot $t$ event is observed, the game manager sends a message
 to each player $i$.  Based on this message, the players choose a   \emph{control action $\alpha_i(t)$}. The random event and the collection of all control actions for slot $t$ determine individual \emph{utilities} $u_i(t)$ for each player $i \in \{1, \ldots, N\}$.  Each player is interested in maximizing the time average of its own utility process.  The game manager is interested in providing messages that lead to a fair allocation of time average utilities across players. 
 
 Specifically, let $\overline{u}_i$ be the time average of $u_i(t)$.  The \emph{fairness} of an achieved vector of time average utilities 
is defined by a \emph{concave fairness function} $\phi(\overline{u}_1, \ldots, \overline{u}_N)$.  
The goal is to devise strategies that maximize $\phi(\overline{u}_1, \ldots, \overline{u}_N)$ subject to certain game-theoretic equilibrium constraints.   
For example, suppose the fairness function is a sum of logarithms:
\[ \phi(\overline{u}_1, \ldots, \overline{u}_N) = \sum_{i=1}^N \log(\overline{u}_i) \]
This corresponds to \emph{proportional fair utility maximization}, a concept often studied in the context of  communication networks \cite{kelly-charging}.  
Another natural concave fairness function is: 
\[ \phi(\overline{u}_1, \ldots, \overline{u}_N) = \min[\overline{u}_1, \ldots, \overline{u}_N, c] \]
for some given constant $c>0$.   This fairness function assigns no added value when the average utility of one player exceeds that of another. 
 
Let $\bv{M}(t) = (M_1(t), \ldots, M_N(t))$ be the \emph{message vector} provided by the game manager on slot $t$.  The value $M_i(t)$ is an element of the set $\script{A}_i$ and represents the action the manager would like player $i$ to take. 
A player $i \in \{1, \ldots, N\}$ is said to \emph{participate} if she always chooses the suggestion of the manager, that is, if 
$\alpha_i(t) = M_i(t)$ for all $t \in \{0, 1, 2, \ldots\}$.   At the beginning of the game, each player makes a participation agreement.  Participating players receive the messages $M_i(t)$, while non-participating players do not. 

This paper considers the class of algorithms that deliver message vectors $\bv{M}(t)$ as a stationary and randomized function of the observed $\bv{\omega}(t)$.   Assuming that all players participate, this induces a conditional probability distribution on the actions, given the current $\bv{\omega}(t)$. The conditional distribution is defined as a \emph{coarse correlated equilibrium} (CCE) if it yields a time average utility vector $(\overline{u}_1, \ldots, \overline{u}_N)$ with the following property \cite{CCE}:  For each player $i \in \{1, \ldots, N\}$, the average utility $\overline{u}_i$  is at least as large as the maximum time average utility this player could achieve if she did not participate (assuming the actions of all other players do not change).    Overall, the goal is to maximize $\phi(\overline{u}_1, \ldots, \overline{u}_N)$ subject to the CCE constraints.  

\subsection{Contributions and related work} 

The notion of \emph{coarse correlated equilibrium} (CCE) was introduced in \cite{CCE} in the static case where there
is no event process $\bv{\omega}(t)$. The CCE definition is 
similar to a \emph{correlated equilibrium}  (CE) \cite{aumann-correlated-eq1}\cite{aumann-correlated-eq2}\cite{game-theory-book}.  The difference is as follows: A correlated equilibrium (CE) is more stringent and requires the utility achieved by each player $i$ to be at least as large as the utility she could achieve if she did not participate \emph{but if she still knew the $M_i(t)$ messages on every slot}.   It is known that both CCE and CE constraints can be written as  linear programs.  Adaptive methods that converge to a CE for static games 
are developed in  \cite{foster-vohra-games-CE}\cite{hart-correlated-games-CE}\cite{fudenberg-games-CE}. 
The concept of \emph{Nash equilibrium} (NE) is more stringent still:  The NE constraint requires all players to act independently and without the aid of a message process $\bv{M}(t)$ \cite{nash-games}\cite{game-theory-book}.  Unfortunately, the problem of computing a Nash equilibrium is nonconvex. 

This paper uses the NE, CE, and CCE concepts in the context of a stochastic game with random 
events $\bv{\omega}(t)$. The optimal action 
associated with a particular event can depend on whether or not the event is rare. This paper develops an online algorithm that is influenced by the event probabilities, but does not require knowledge of these probabilities.   
The algorithm uses the Lyapunov optimization theory of \cite{sno-text}\cite{now} and is of the \emph{max-weight} type.  Specifically, every slot $t$, the game manager observes the $\bv{\omega}(t)$ realization and chooses a suggestion vector  
by greedily minimizing a \emph{drift-plus-penalty expression}.  Such Lyapunov methods are used
extensively in the context of queueing
 networks \cite{tass-server-allocation}\cite{neely-fairness-ton} (see also related methods in 
 \cite{atilla-fairness-ton}\cite{stolyar-gpd-gen}\cite{tutorial-lin}).    This is perhaps the first use of such techniques in a game-theoretic setting.   
 
One reason the solution of this paper can have a simple structure is that the random event process $\bv{\omega}(t)$ is assumed to be independent of the prior control actions.  Specifically, while the components $\omega_i(t)$ are allowed to be arbitrarily correlated across 
$i \in \{0, 1, \ldots, N\}$, the vector $\bv{\omega}(t)$ is assumed to be 
independent and identically distributed (i.i.d.) over slots.  Prior work on stochastic games considers more complex problems where 
$\bv{\omega}(t+1)$ is influenced by the control action of slot $t$, including work in \cite{cor-eq-stochastic-games} which studies correlated equilibria in this context.  This typically involves Markov decision theory and has high complexity. 
Specifically, if $\Omega_i$ is the set of all possible values of $\omega_i(t)$, and if $|\Omega_i|$ is the (finite) size of this set, then complexity is typically \emph{at least} as large as $\prod_{i=1}^N|\Omega_i|$. 

In contrast, while the current paper treats a stochastic problem with more limited structure, the resulting solution is simple and grows as $\sum_{i=1}^N |\Omega_i|$. Specifically, the algorithm uses a number of \emph{virtual queues} that is linear in $N$, rather than exponential in $N$, resulting in polynomial bounds on convergence time.   
Furthermore, the number of virtual queues grows only linearly in the size of each set $\Omega_i$.  This improves on the original conference version of this paper \cite{repeated-games-allerton}, which 
required a number of virtual queues that was exponential in the size of $\Omega_i$.  The exponential-to-polynomial 
improvement is done by equivalently modeling the constraints via a grouping of  conditional expectations given an observed random event.

 \section{Static games} \label{section:static-games} 
 
 This section introduces the problem in the \emph{static case} without random processes $\omega_0(t), \omega_1(t), \ldots, \omega_N(t)$.  The different forms of equilibrium are defined and compared through a simple example. 
   The general stochastic problem is treated in Section \ref{section:stochastic-games}. 
   
   Suppose there are $N$ players, where $N$ is an integer larger than 1.  Each player $i \in \{1, \ldots, N\}$ 
   has an \emph{action space} $\script{A}_i$, assumed to be a finite set.   The game operates in slotted time $t \in \{0, 1, 2, \ldots\}$. 
   Every slot $t$, each player $i$ chooses an action $\alpha_i(t) \in \script{A}_i$.  Let $\bv{\alpha}(t) = (\alpha_1(t), \ldots, \alpha_N(t))$ 
   be the vector of control actions on slot $t$.  The utility $u_i(t)$ earned by player $i$ on slot
   $t$ is a real-valued function of $\bv{\alpha}(t)$: 
   \[ u_i(t) = \hat{u}_i(\bv{\alpha}(t)) \: \: \forall i \in \{1, \ldots, N\}   \]
   The utility functions $\hat{u}_i(\bv{\alpha})$ can be different for each player $i$.    Define $\script{A} = \script{A}_1  \times \cdots \times \script{A}_N$.   Consider starting with a particular vector $\bv{\alpha} \in \script{A}$ and  modifying it by changing a single entry $i$ from $\alpha_i$ to some other action $\beta_i$. This new vector is represented by the notation $(\beta_i, \bv{\alpha}_{\overline{i}})$. 
 Define $\script{A}_{\overline{i}}$ as the set of all vectors $\bv{\alpha}_{\overline{i}}$, being the set product of 
 $\script{A}_j$ over all  $j \neq i$.

   The three different forms of equilibrium considered in this section are defined by probability mass functions $Pr[\bv{\alpha}]$ for $\bv{\alpha} \in \script{A}$. It is assumed throughout that: 
   \begin{itemize} 
   \item  $Pr[\bv{\alpha}] \geq 0$ for all $\bv{\alpha} \in \script{A}$.
   \item  $\sum_{\bv{\alpha}\in\script{A}} Pr[\bv{\alpha}] = 1$. 
   \end{itemize} 
    If actions $\bv{\alpha}(t)$ are chosen independently every slot according to the same probability mass function $Pr[\bv{\alpha}]$, the law of large numbers ensures that, with probability 1, the time average utility of each player $i \in \{1, \ldots, N\}$ 
is: 
\[ \mbox{$\overline{u}_i = \sum_{\bv{\alpha}\in\script{A}} Pr[\bv{\alpha}]\hat{u}_i(\bv{\alpha})$} \]

\subsection{Nash equilibrium (NE)} 

The standard concept of Nash equilibrium from \cite{nash-n-person-games}\cite{nash-games}
assumes players take independent actions, so that: 
\begin{equation} \label{eq:NE-ind} 
Pr[\bv{\alpha}] = \prod_{i=1}^{N}Pr[\alpha_i(t) = \alpha_i] 
\end{equation} 
A probability mass function $Pr[\bv{\alpha}]$ is a 
\emph{mixed strategy Nash equilibrium (NE)} if it satisfies: 
\begin{eqnarray}
\sum_{\bv{\alpha} \in \script{A}} Pr[\bv{\alpha}]\hat{u}_i(\bv{\alpha}) \geq \sum_{\bv{\alpha}\in\script{A}} Pr[\bv{\alpha}] \hat{u}_i(\beta_i, \bv{\alpha}_{\overline{i}}) \nonumber \\
\forall i \in \{1, \ldots, N\}, \forall \beta_i \in \script{A}_i \label{eq:nash-static} 
\end{eqnarray}

\subsection{Correlated equilibrium (CE)} 

The standard concept of correlated equilibrium from \cite{aumann-correlated-eq1}\cite{aumann-correlated-eq2} can be motivated by a game manager that provides suggested
actions $(\alpha_1(t), \ldots, \alpha_N(t))$ every slot $t$, where player $1$ only sees $\alpha_1(t)$, player $2$ only sees $\alpha_2(t)$, and so on.  Assume the suggestion vector is  independent and identically distributed (i.i.d.)  over slots with some probability mass function $Pr[\bv{\alpha}]$.
Assume all players participate, so that every slot their chosen actions match the suggestions.  The probability mass function $Pr[\bv{\alpha}]$ is a \emph{correlated equilibrium (CE)} if:  
 \begin{align}
 &\sum_{\bv{\alpha}_{\overline{i}}\in\script{A}_{\overline{i}}} Pr[\alpha_i, \bv{\alpha}_{\overline{i}}] \hat{u}_i(\alpha_i, \bv{\alpha}_{\overline{i}}) \geq   \sum_{\bv{\alpha}_{\overline{i}}\in\script{A}_{\overline{i}}} Pr[\alpha_i, \bv{\alpha}_{\overline{i}}]\hat{u}_i(\beta_i, \bv{\alpha}_{\overline{i}}) \nonumber \\
&\forall i \in \{1, \ldots, N\}, \forall \alpha_i \in \script{A}_i, \forall \beta_i \in \script{A}_i  \mbox{ with $\beta_i \neq \alpha_i$}  \label{eq:ce-static}
 \end{align}

 These constraints imply that no player can gain a larger average utility by individually 
 deviating from the suggestions of the game manager \cite{aumann-correlated-eq2}.  This can be understood as follows:   Fix an $i \in \{1, \ldots, N\}$ and an $\alpha_i \in \script{A}_i$ such that $Pr[\alpha_i(t)=\alpha_i]>0$. Divide both sides of the above inequality by $Pr[\alpha_i(t)=\alpha_i]$. Then: 
 \begin{itemize} 
 \item The left-hand-side is the conditional expected utility of player $i$, given that all players participate and that player $i$ 
 sees suggestion $\alpha_i$ on the current slot.  
 \item  The right-hand-side is the conditional expected utility of player $i$, given that she sees $\alpha_i$ on the current slot, that all other players $j \neq i$ participate, and that player $i$ chooses action $\beta_i$ instead of $\alpha_i$ (so player $i$ does \emph{not} participate). 
 \end{itemize} 
 
  The correlated equilibrium constraints are linear in the $Pr[\bv{\alpha}]$ variables.  
Define $|\script{A}_i|$ as the number of actions in set $\script{A}_i$.  The number of linear constraints specified by \eqref{eq:ce-static} is then: 
\begin{equation} \label{eq:number-constraints-CE} 
\mbox{$\sum_{i=1}^N |\script{A}_i|(|\script{A}_i|-1)$}  
\end{equation} 

\subsection{Coarse correlated equilibrium (CCE)} 

The definition of correlated equilibrium assumes that non-participating players still receive the suggestions from the 
game manager.  As the suggestion $\alpha_i(t)$ for player $i$ may be correlated with the suggestions $\alpha_j(t)$ of other players $j \neq i$, this can give a non-participating player $i$ a great deal of 
information about the likelihood of actions from other players.  
The following simple modification assumes that non-participating players do not receive any suggestions from the game manager.   A probability mass function $Pr[\bv{\alpha}]$ is a \emph{coarse correlated equilibrium (CCE)} if it 
satisfies the constraints \eqref{eq:nash-static}.  Note that the product form constraints \eqref{eq:NE-ind} are not required. 
 This CCE definition was introduced in \cite{CCE}.   
 Similar to the CE case, these CCE constraints imply that no player can increase her average utility by individually deviating from the suggestions of the game manager.

The CCE constraints \eqref{eq:nash-static} 
are linear in the $Pr[\bv{\alpha}]$ values.  The number of CCE constraints is: 
\[ \mbox{$\sum_{i=1}^N |\script{A}_i|$}   \]
This number is typically much less than the number of constraints required for a CE,  specified in \eqref{eq:number-constraints-CE}.  Assuming that $|\script{A}_i| \geq 2$ for each player $i$ (so that each player has at least 2 action options), the number of CCE constraints is always less than or equal to the number of CE constraints, with equality if and only if $|\script{A}_i|=2$ for all players $i$.

\subsection{A superset result} 

The assumption that all sets $\script{A}_i$ are finite make the game a \emph{finite game}. 
Fix a finite game and define $\script{E}_{NE}$, $\script{E}_{CE}$, and $\script{E}_{CCE}$ as the set of all probability mass functions $Pr[\bv{\alpha}]$ that define a (mixed strategy) Nash equilibrium, a correlated equilibrium, and a coarse correlated equilibrium, respectively.  It is known that every such finite game has at least one  mixed strategy Nash equilibrium, and so $\script{E}_{NE}$ is nonempty \cite{nash-n-person-games}\cite{nash-games}.  
Furthermore, it is known that any NE is also a CE, and any CE is also a CCE, so that \cite{aumann-correlated-eq1}\cite{aumann-correlated-eq2}\cite{CCE}: 
\begin{equation} \label{eq:inclusion-static} 
\script{E}_{NE} \subseteq \script{E}_{CE} \subseteq \script{E}_{CCE} 
\end{equation} 
Furthermore, the sets  $\script{E}_{CE}$ and $\script{E}_{CCE}$ are closed, bounded, and convex \cite{aumann-correlated-eq1}\cite{aumann-correlated-eq2}\cite{CCE}.   



\subsection{A simple example} \label{section:example} 

Consider a game where player $1$ has three control options and player $2$ has two control options: 
\[ \script{A}_1 = \{\alpha, \beta, \gamma\} \: \: , \: \: \script{A}_2 = \{\alpha, \beta\} \]
The utility functions $\hat{u}_1(\alpha_1, \alpha_2)$ and $\hat{u}_2(\alpha_1, \alpha_2)$ are specified in the table of Fig. \ref{fig:example-utils}, where player 1 actions are listed by row and player 2 actions are listed by column.  

\begin{figure}[htbp]
   \centering
  \[\begin{array}{ccc}
\mbox{Utility 1} & \mbox{Utility 2} & \mbox{Probabilities} \\
\begin{tabular}{|c|c|c|}
\hline 
& $\alpha$ & $\beta$ \\ \hline 
$\alpha$ & 2 & 5  \\ \hline
$\beta$ & 4 & 2 \\ \hline
$\gamma$ & 3 & 5 \\ \hline
\end{tabular} & 
\begin{tabular}{|c|c|c|}
\hline 
& $\alpha$ & $\beta$ \\ \hline 
$\alpha$ & 50 & 1  \\ \hline
$\beta$ & 2 & 4 \\ \hline
$\gamma$ & 3 & 0 \\ \hline
\end{tabular} & 
\begin{tabular}{|c|c|c|}
\hline 
& $\alpha$ & $\beta$ \\ \hline 
$\alpha$ & $a$ & $b$  \\ \hline
$\beta$ & $c$ & $d$ \\ \hline
$\gamma$ & $e$ & $f$ \\ \hline
\end{tabular}
 \end{array}\] 
   \caption{Example utility functions $\hat{u}_1(\alpha_1, \alpha_2)$ and $\hat{u}_2(\alpha_1, \alpha_2)$.}
   \label{fig:example-utils}
\end{figure}

There are six possible action vectors $(\alpha_1, \alpha_2)$.  Define the mass function $Pr[\bv{\alpha}]$ by values $a$, $b$, $c$, $d$, $e$, $f$ associated with each of the six possibilities, as shown in Fig. \ref{fig:example-utils}. 

The eight CE constraints for this problem are: 
\begin{eqnarray*}
\mbox{player 1 sees $\alpha$:} & 2a + 5b \geq 4a + 2b \\
\mbox{player 1 sees $\alpha$:} & 2a + 5b \geq 3a + 5b \\
\mbox{player 1 sees $\beta$:} & 4c + 2d \geq 2c + 5d \\
\mbox{player 1 sees $\beta$:} & 4c + 2d \geq 3c + 5d \\
\mbox{player 1 sees $\gamma$:} & 3e + 5f \geq 2e + 5f \\
\mbox{player 1 sees $\gamma$:} & 3e + 5f \geq 4e + 2f \\
\mbox{player 2 sees $\alpha$:} & 50a + 2c + 3e \geq a + 4c + 0e \\
\mbox{player 2 sees $\beta$:} &  b + 4d + 0f \geq 50b + 2d + 3f 
\end{eqnarray*}
It can be shown that there is a single probability mass function $Pr[\bv{\alpha}]$ that satisfies all of these CE constraints: 
\[ a=b=0, \: \: c = 0.45, \: \: d = 0.15, \: \: e = 0.3, \: \: f = 0.1 \]
This is also the only NE.  The average utility vector associated with this mass function is
$(\overline{u}_1, \overline{u}_2) = (3.5, 2.4)$. 

In contrast, the five CCE constraints for this problem are: 
\begin{eqnarray*}
\mbox{player 1 chooses $\alpha$:} & \hspace{-.4in}2a + 5b + 4c + 2d + 3e + 5f  \\
&\hspace{+.3in} \geq 2(a+c+e) + 5(b+d+f) \\
\mbox{player 1 chooses  $\beta$:} & \hspace{-.4in}2a + 5b + 4c + 2d + 3e + 5f \\
& \hspace{+.3in} \geq 4(a+c+e) + 2(b+d+f) \\
\mbox{player 1 chooses  $\gamma$:} & \hspace{-.4in}2a + 5b + 4c + 2d + 3e + 5f \\
& \hspace{+.3in} \geq 3(a+c+e) + 5(b+d+f) \\
\mbox{player 2 chooses $\alpha$:} & \hspace{-.7in} 50a + b + 2c + 4d + 3e  \\
&\hspace{+.2in} \geq 50(a+b) + 2(c+d) +3(e+f) \\
\mbox{player 2 chooses $\beta$:} &  \hspace{-.7in} 50a + b + 2c + 4d + 3e  \\
&\hspace{+.2in} \geq 1(a+b) + 4(c+d) + 0(e+f)
\end{eqnarray*}
There are an infinite number of probability mass functions $Pr[\bv{\alpha}]$ that satisfy these CCE constraints.  Three different ones are given in the table of Fig. \ref{fig:pce-solns}, labeled \emph{distribution 1}, \emph{distribution 2}, 
 and \emph{distribution 3}. 
Distribution 1 corresponds to the CE and NE distribution.

The set of all utility vectors $(\overline{u}_1, \overline{u}_2)$ achievable under CCE constraints is the triangular region  shown in Fig. \ref{fig:CCE-compare}.  The three vertices of the triangle correspond to the three distributions in Fig. \ref{fig:pce-solns}, and are: 
\[ (\overline{u}_1, \overline{u}_2) \in \{(3.5, 2.4), (3.5, 9.3),  (3.8773, 3.7914)\}  \]

The point $(3.5, 2.4)$ is the lower left vertex of the triangle and corresponds to the CE (and NE) distribution.  It is clear that both players can significantly increase their utility by changing from CE constraints to CCE constraints.   This illustrates the following general principle:  \emph{All players benefit if non-participants are denied access to the suggestions of the game manager}.  This principle is justified by \eqref{eq:inclusion-static}.

\begin{figure}[htbp]
   \centering
  \[\begin{array}{ccc}
\mbox{Distribution 1} & \mbox{Distribution 2} & \mbox{Distribution 3} \\
\begin{tabular}{|c|c|c|}
\hline 
& $\alpha$ & $\beta$ \\ \hline 
$\alpha$ & 0 & 0  \\ \hline
$\beta$ & .45 & .15 \\ \hline
$\gamma$ & .30 & .10 \\ \hline
\end{tabular} & 
\begin{tabular}{|c|c|c|}
\hline 
& $\alpha$ & $\beta$ \\ \hline 
$\alpha$ & .15 & 0  \\ \hline
$\beta$ & .60 & .15 \\ \hline
$\gamma$ & 0 & .10 \\ \hline
\end{tabular} & 
\begin{tabular}{|c|c|c|}
\hline 
& $\alpha$ & $\beta$ \\ \hline 
$\alpha$ & .0368 & 0  \\ \hline
$\beta$ & .9018 & .0368 \\ \hline
$\gamma$ & 0 & .0245 \\ \hline
\end{tabular}
 \end{array}\] 
   \caption{Three different probability distributions that satisfy the CCE constraints. The first distribution also satisfies the CE and NE constraints.}
   \label{fig:pce-solns}
\end{figure}

\begin{figure}[htbp]
   \centering
   \includegraphics[width=3.8in]{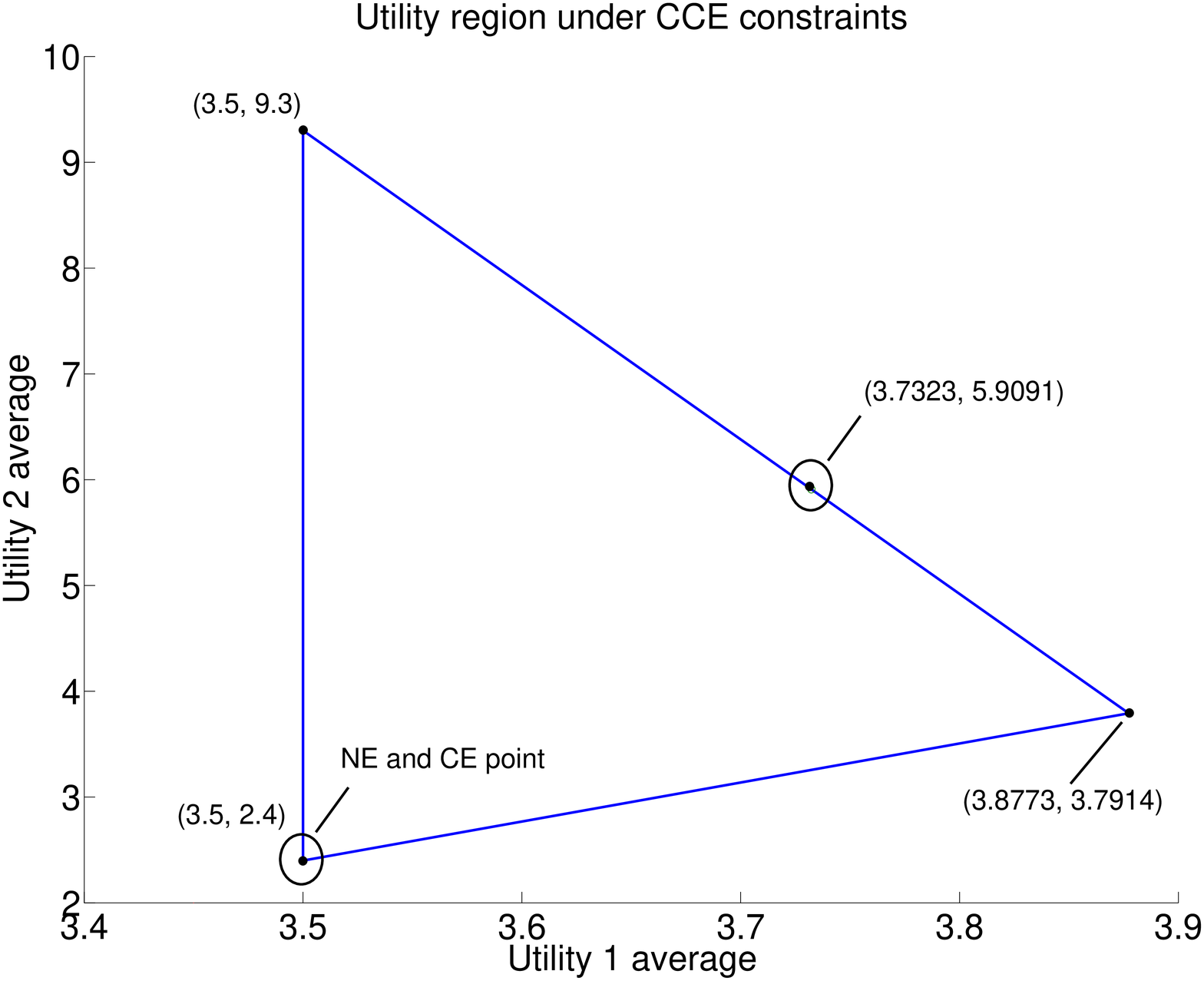} 
   \caption{The region of $(\overline{u}_1, \overline{u}_2)$ values achievable under CCE constraints. All points inside and on the triangle are achievable.  The NE and CCE point is the lower left vertex. The point $(3.7323, 5.9091)$ is the solution to the convex optimization example of Section \ref{section:convex-static}.}
   \label{fig:CCE-compare}
\end{figure}

\subsection{Utility optimization with equilibrium constraints} \label{section:convex-static}

There are typically many  probability distributions $Pr[\bv{\alpha}]$ that satisfy the CCE constraints. The goal is to find one that leads to an optimal vector of average utilities.  Optimality is determined by a concave \emph{fairness function}, as defined below.

For convenience, assume all utility functions are nonnegative.  Define $u_i^{max}$ as an upper bound on the utility for each player $i \in \{1, \ldots, N\}$, so that: 
\[ 0 \leq \hat{u}_i(\bv{\alpha}) \leq u_i^{max} \: \: \: \forall \bv{\alpha} \in \script{A} \]
 Define $\phi(u_1, \ldots, u_N)$ as a continuous and concave function that maps the set $\times_{i=1}^N[0, u_i^{max}]$ to the real numbers.  This is called the \emph{fairness function}.  The game manager chooses a probability mass function $Pr[\bv{\alpha}]$ with the goal of maximizing $\phi(\overline{u}_1, \ldots, \overline{u}_N)$ subject to CCE constraints: 
 
 \begin{eqnarray}
 \mbox{Maximize:} & \phi(\overline{u}_1, \ldots, \overline{u}_N) \label{eq:ps1} \\
 \mbox{Subject to:} & \overline{u}_i = \sum_{\bv{\alpha}\in\script{A}} Pr[\bv{\alpha}]\hat{u}_i(\bv{\alpha}) \: \: \forall i \in \{1, \ldots, N\} \label{eq:ps2} \\
 & Pr[\bv{\alpha}] \geq 0 \: \: \forall \bv{\alpha} \in \script{A} \label{eq:ps3} \\
 & \sum_{\bv{\alpha}\in\script{A}} Pr[\bv{\alpha}] = 1 \label{eq:ps4} \\
 & \mbox{CCE constraints \eqref{eq:nash-static} are satisfied} \label{eq:ps5} 
 \end{eqnarray}
 
 The above is a convex optimization problem.    If the CCE constraints are replaced by the CE constraints \eqref{eq:ce-static}, the problem remains convex but can have significantly more constraints. If the CCE constraints are replaced with the NE constraints \eqref{eq:NE-ind}-\eqref{eq:nash-static}, the problem becomes nonconvex. 
 
 Consider the special case example of Section \ref{section:example} with fairness function given by: 
  \[ \phi(\overline{u}_1, \overline{u}_2) = 10\log(1 + \overline{u}_1) + \log(1 + \overline{u}_2) \]
  where player 1 is given a higher priority.  The optimal utility is $(\overline{u}_1^*, \overline{u}_2^*) = (3.7323, 5.9091)$, plotted in Fig. \ref{fig:CCE-compare}. 
 
  
\section{Stochastic games} \label{section:stochastic-games} 

Let $\bv{\omega}(t) = (\omega_0(t), \omega_1(t), \ldots, \omega_N(t))$ be a vector of random events for 
slot $t \in \{0, 1, 2, \ldots\}$.  Each component $\omega_i(t)$ takes values in some finite set $\Omega_i$, 
for $i \in \{0, 1, \ldots, N\}$. 
Define 
$\Omega = \Omega_0 \times \Omega_1 \times \cdots \times \Omega_N$. 
The vector process $\bv{\omega}(t)$ is assumed to be independent and identically distributed (i.i.d.) over slots with probability 
mass function: 
\[ \pi[\bv{\omega}] \defequiv Pr[\bv{\omega}(t) = \bv{\omega}] \: \: \: \forall \bv{\omega} \in \Omega \]
where the notation ``$\defequiv$'' means ``defined to be equal to.''
On each slot $t$, 
the components of the vector $\bv{\omega}(t)$ can be arbitrarily correlated. 

At the beginning of each slot $t$, each player $i \in \{1, \ldots, N\}$ observes its own random event $\omega_i(t)$.  The game manager observes the full vector $\bv{\omega}(t)$, including the additional information $\omega_0(t)$.  It then sends a suggested action $M_i(t)$ to each participating player $i \in \{1, \ldots, N\}$.  Assume $M_i(t) \in \script{A}_i$, where $\script{A}_i$ is the finite set of actions available to player $i$. 
Each player $i$ chooses an action $\alpha_i(t) \in \script{A}_i$.  Participating players always choose $\alpha_i(t) = M_i(t)$.  Non-participating players do not receive $M_i(t)$ and choose $\alpha_i(t)$ using knowledge of only
 $\omega_i(t)$ and of events that occurred before slot $t$. 

Let $\bv{\alpha}(t) = (\alpha_1(t), \ldots, \alpha_N(t))$ be the action vector.  
The utility $u_i(t)$ earned by each player $i$ on slot $t$ is a function of $\bv{\alpha}(t)$ and $\bv{\omega}(t)$: 
\[ u_i(t) = \hat{u}_i(\bv{\alpha}(t), \bv{\omega}(t)) \]
For convenience, assume utility functions are nonnegative with maximum values $u_i^{max}$ for $i \in \{1, \ldots, N\}$, 
so that: 
\[ 0 \leq \hat{u}_i(\bv{\alpha}(t), \bv{\omega}(t)) \leq u_i^{max} \]

\subsection{Discussion of game structures} 

This model can be used to treat various game structures.   For example, the scenario where all players have full information can be treated by defining $\omega_i(t) = \omega_0(t)$ for all $i \in \{1, \ldots, N\}$.  This is useful in games related to economic markets, 
where  $\omega_0(t)$ can represent a commonly known vector of current prices.   Alternatively, one can imagine a game 
with a single random event process $\omega_0(t)$ that is known to the game manager but unknown to all players. For example, 
consider a game defined over a wireless multiple access system.  Wireless users are players in the game, and the access point is the game manager.  In this example, $\omega_0(t)$ can represent a vector of current channel conditions known only to the access point.   Such games can be treated by setting $\omega_i(t)$ to a default constant value for all  $i \in \{1, \ldots, N\}$ 
 and all slots $t$.

 \subsection{Pure strategies and the virtual static game} 
 
 Assume all players participate, so that $M_i(t) = \alpha_i(t)$ for all $i$.  For each $i \in \{1, \ldots, N\}$, denote the sizes of sets $\Omega_i$ and $\script{A}_i$ by $|\Omega_i|$ and $|\script{A}_i|$, respectively. Define a \emph{pure strategy function for player $i$} as a function $b_i(\omega_i)$ that maps $\Omega_i$ to the set $\script{A}_i$.   There are $|\script{A}_i|^{|\Omega_i|}$ such functions.  Define: 
 \[ \script{S}_i \defequiv \{1, 2, \ldots, |\script{A}_i|^{|\Omega_i|}\} \]
Enumerate the pure strategy functions for player $i$ and represent them by  $b_i^{(s)}(\omega_i)$ for $s \in\script{S}_i$.  
 Define: 
 \[ \script{S} \defequiv \script{S}_1 \times \script{S}_2 \times \cdots \times \script{S}_N \]
 Each vector $(s_1, s_2, \ldots, s_N) \in \script{S}$ can be used to specify a profile of pure strategies used by each player.  
 For each  $\bv{s} \in \script{S}$ and each  $\bv{\omega} \in \Omega$,  define: 
 \begin{equation} \label{eq:structure-b-functions} 
 \bv{b}^{(\bv{s})}(\bv{\omega}) = (b_1^{(s_1)}(\omega_1), b_2^{(s_2)}(\omega_2), \ldots, b_N^{(s_N)}(\omega_N)) 
 \end{equation} 
In the special case when the action of each player $i$ on slot $t$ is defined by pure strategy 
 $s_i$, the action vector is $(\alpha_1(t), \ldots, \alpha_N(t))=\bv{b}^{(\bv{s})}(\bv{\omega}(t))$. 
 The average utility earned by player $i$  on such a slot $t$  is defined: 
 \begin{equation} \label{eq:hi} 
  \mbox{$h_i(\bv{s}) \defequiv \sum_{\bv{\omega}\in\Omega} \pi[\bv{\omega}]\hat{u}_i(\bv{b}^{(\bv{s})}(\bv{\omega}), \bv{\omega})$}
  \end{equation}

 The stochastic game can be treated as a \emph{virtual static game} as follows:  The virtual static game also has $N$ players.  The \emph{virtual action space} of each player $i$ is viewed as the set of pure strategies $\script{S}_i$.  Every slot $t$, each player $i$ selects a pure strategy $s_i(t) \in  \script{S}_i$.   The \emph{virtual utility functions} are given by the functions $h_i(\bv{s})$.   
 
 The virtual static game is still a finite game.  Hence, the NE, CE, and CCE definitions for static games can be used here.  In particular, let $Pr[\bv{s}]$ be a probability mass function over the finite set of strategy profiles $\bv{s} \in \script{S}$.   Then: 
 
 \begin{itemize} 
 \item (NE for virtual static game)  $Pr[\bv{s}]$ is a NE for the virtual static game if it has the product form: 
 \begin{equation} \label{eq:product-form-virtual} 
  Pr[\bv{s}] = \prod_{i=1}^N g_i[s_i] \: \: \: \: \forall \bv{s} \in \script{S}  
  \end{equation} 
 where $g_i[s_i] = Pr[s_i(t) = s_i]$, and if: 
 \begin{eqnarray} 
 \sum_{\bv{s} \in \script{S}} Pr[\bv{s}]h_i(\bv{s}) \geq \sum_{\bv{s}\in\script{S}} Pr[\bv{s}]h_i(r_i, \bv{s}_{\overline{i}}) \nonumber \\   
 \forall i \in \{1, \ldots, N\} , \forall r_i \in \script{S}_i  \label{eq:NE-virtual} 
 \end{eqnarray}
 
 \item (CE for virtual static game) $Pr[\bv{s}]$ is a CE for the virtual static game if: 
  \begin{eqnarray} 
  \sum_{\bv{s}_{\overline{i}}\in \script{S}_{\overline{i}}} Pr[s_i, \bv{s}_{\overline{i}}]h_i(s_i, \bv{s}_{\overline{i}}) \geq
  \sum_{\bv{s}_{\overline{i}} \in \script{S}_{\overline{i}}} Pr[s_i, \bv{s}_{\overline{i}}]h_i(r_i, \bv{s}_{\overline{i}}) \nonumber \\
  \forall i \in \{1, \ldots, N\}, \forall s_i, r_i \in \script{S}_i \label{eq:CE-virtual} 
  \end{eqnarray}
  
  \item (CCE for virtual static game) $Pr[\bv{s}]$ is a CCE for the virtual static game if it satisfies \eqref{eq:NE-virtual}. 
 \end{itemize} 
 
 A given probability mass function $Pr[\bv{s}]$ defined over $\bv{s} \in \script{S}$ generates a 
 conditional probability mass function $Pr[\bv{\alpha}|\bv{\omega}]$ defined over all $\bv{\alpha} \in \script{A}$ and $\bv{\omega} \in \Omega$: 
 \begin{equation} \label{eq:generate} 
 \mbox{$Pr[\bv{\alpha}|\bv{\omega}] = \sum_{\bv{s} \in \script{S}} Pr[\bv{s}]1\{\bv{b}^{(\bv{s})}(\bv{\omega})=\bv{\alpha}\}$}
 \end{equation} 
 where $1\{\bv{b}^{(\bv{s})}(\bv{\omega})=\bv{\alpha}\}$ is an indicator function that is 1 if $\bv{b}^{(\bv{s})}(\bv{\omega})=\bv{\alpha}$, and is 0 else.    
 However, not all 
 $Pr[\bv{\alpha}|\bv{\omega}]$ functions can be generated in this way.\footnote{The conference version of this paper \cite{repeated-games-allerton} contained an incorrect statement suggesting that all $Pr[\bv{\alpha}|\bv{\omega}]$ distributions can be generated by $Pr[\bv{s}]$ distributions according to \eqref{eq:generate} (Lemma 1 from page 5 of \cite{repeated-games-allerton}).  While this is true in the case when   $\bv{b}^{(\bv{s})}(\bv{\omega})$ can be an arbitrary function of the full $\bv{\omega}$ vector, it  
 does \emph{not} hold for  strategy functions with the structure \eqref{eq:structure-b-functions}. The author regrets the misleading statement in \cite{repeated-games-allerton}.  Fortunately, that incorrect statement was never used, and so  
  it did not affect any of the results in \cite{repeated-games-allerton}.}
 
 For example, the right-hand-side of \eqref{eq:generate} does not depend on $\omega_0$.  In contrast, a game manager might want to select $Pr[\bv{\alpha}|\bv{\omega}]$ as a function of the full random event vector $\bv{\omega} = (\omega_0, \omega_1, \ldots, \omega_N)$.  
 
 If \eqref{eq:generate} holds, 
a game manager with no knowledge of the random event vector $\bv{\omega}(t)$ could 
produce suggestions according to  $Pr[\bv{\alpha}|\bv{\omega}]$ by randomly selecting a strategy vector $\bv{s} = (s_1, \ldots, s_N)$ with probability $Pr[\bv{s}]$, and then broadcasting component $s_i$ to each player $i$.  Thus, the NE, CE, and CCE conditions in \eqref{eq:product-form-virtual}-\eqref{eq:CE-virtual} for the virtual static game 
can be viewed as \emph{information restricted (IR)} notions of equilibrium for the stochastic game.   Formally,
define a probability mass function $Pr[\bv{s}]$ to be an IR-NE if it satisfies \eqref{eq:product-form-virtual}-\eqref{eq:NE-virtual}, an IR-CE if it satisfies \eqref{eq:CE-virtual}, and an IR-CCE if it satisfies \eqref{eq:NE-virtual}.

  \begin{lem} \label{lem:product-form} 
 Suppose $Pr[\bv{s}]$ and $Pr[\bv{\alpha}|\bv{\omega}]$ satisfy \eqref{eq:generate}.  If $Pr[\bv{s}]$ has the product
 form \eqref{eq:product-form-virtual}, then $Pr[\bv{\alpha}|\bv{\omega}]$ has the following product form: 
 \begin{eqnarray}
 Pr[\bv{\alpha}|\bv{\omega}] &=& \prod_{i=1}^NPr[\alpha_i|\omega_i] \: \: \forall \bv{\omega} \in \Omega, \forall \bv{\alpha} \in \script{A}  \label{eq:product-form-conditional} 
 \end{eqnarray}
 \end{lem} 
 
 \begin{proof} 
It follows by \eqref{eq:generate} that: 
\begin{eqnarray*}
Pr[\bv{\alpha}|\bv{\omega}] &=& 
\sum_{\bv{s}\in\script{S}}\prod_{i=1}^Ng_i[s_i]1\{\bv{b}^{(\bv{s})}(\bv{\omega}) = \bv{\alpha}\} \\
&=& \sum_{\bv{s}\in\script{S}} \prod_{i=1}^Ng_i[s_i]1\{b_i^{(s_i)}(\omega_i) = \alpha_i\} \\
&=&\left(\sum_{s_1\in\script{S}_1}  g_1[s_1]1\{b_1^{(s_1)}(\omega_1) = \alpha_1\}\right)\\
&& \: \: \: \: \cdots \left(\sum_{s_N\in\script{S}_N}  g_N[s_N]1\{b_N^{(s_N)}(\omega_N) = \alpha_N\}\right) \\
&=&Pr[\alpha_1|\omega_1] \cdots Pr[\alpha_N|\omega_N]
\end{eqnarray*} 
\end{proof} 

\subsection{General equilibrium for the stochastic game}

 Let $Pr[\bv{\alpha}|\bv{\omega}]$ be a conditional probability mass function defined over $\bv{\omega} \in \Omega$, $\bv{\alpha} \in \script{A}$.  It is assumed throughout that: 
 \begin{eqnarray} 
 \mbox{$Pr[\bv{\alpha}|\bv{\omega}] \geq 0$} &  \forall \bv{\alpha}\in\script{A}, \forall \bv{\omega} \in \Omega \label{eq:dist1} \\
 \mbox{$\sum_{\bv{\alpha}\in\script{A}}Pr[\bv{\alpha}|\bv{\omega}] = 1$} & \forall \bv{\omega} \in \Omega \label{eq:dist2} 
 \end{eqnarray} 
   

General equilibria for the \emph{stochastic game} can be defined in terms of  $Pr[\bv{\alpha}|\bv{\omega}]$. 
The conference version of this paper  \cite{repeated-games-allerton} does this by specifying constraints for each pure strategy $r_i \in \script{S}_i$, similar to the 
virtual static game constraints \eqref{eq:NE-virtual} and \eqref{eq:CE-virtual}. Unfortunately, this requires a number of constraints that is exponential in the size of the sets $\Omega_i$.  The following alternative definition is equivalent 
to that given in \cite{repeated-games-allerton}, 
yet uses only a polynomial number of constraints.  For the case of NE and CCE, it 
does so by introducing additional variables $\theta_i(v_i)$ for each $i \in \{1, \ldots, N\}$ and each $v_i \in \Omega_i$.   
Intuitively, $\theta_i(v_i)$ represents the  largest conditional expected utility achievable by player $i$, given that she does not participate and that she observes $\omega_i(t) = v_i$.   

\begin{itemize} 
\item (NE for the stochastic game) $Pr[\bv{\alpha}|\bv{\omega}]$ is a NE for the stochastic game if it has the product form
\eqref{eq:product-form-conditional} and if there are real numbers $\theta_i(v_i) \in [0, u_i^{max}]$ such that: 
\begin{align}
&\sum_{\bv{\omega} \in \Omega}\sum_{\bv{\alpha}\in\script{A}} \pi[\bv{\omega}] Pr[\bv{\alpha}|\bv{\omega}]\hat{u}_i(\bv{\alpha}, \bv{\omega})  \nonumber \\
&\geq \sum_{\bv{\omega} \in \Omega}\sum_{\bv{\alpha}\in\script{A}} \pi[\bv{\omega}] Pr[\bv{\alpha}|\bv{\omega}] \theta_i(\omega_i) \: \: \forall i \in \{1, \ldots, N\}  \label{eq:NE-better1} 
\end{align} 
and 
\begin{align} 
 &\sum_{\bv{\omega} \in \Omega | \omega_i = v_i}\sum_{\bv{\alpha}\in\script{A}} \pi[\bv{\omega}] Pr[\bv{\alpha}|\bv{\omega}] 
\theta_i(v_i) \nonumber \\
& \geq \sum_{\bv{\omega}\in\Omega| \omega_i=v_i} \sum_{\bv{\alpha}\in\script{A}} \pi[\bv{\omega}]Pr[\bv{\alpha}|\bv{\omega}]   \hat{u}_i((\beta_i, \bv{\alpha}_{\overline{i}}), \bv{\omega}) \nonumber \\
 &\hspace{+.3in} \forall i \in \{1, \ldots, N\} , \forall v_i \in \Omega_i, \forall \beta_i \in \script{A}_i \label{eq:NE-better2} 
\end{align}

\item (CE for the stochastic game) $Pr[\bv{\alpha}|\bv{\omega}]$ is a CE for the stochastic game if there are real numbers 
$\theta_i(v_i, c_i) \in [0, u_i^{max}]$ such that: 
\begin{align}
&\sum_{\bv{\omega} \in \Omega}\sum_{\bv{\alpha}\in\script{A}} \pi[\bv{\omega}]Pr[\bv{\alpha}|\bv{\omega}][\hat{u}_i(\bv{\alpha}, \bv{\omega}) - \theta_i(\omega_i, \alpha_i)]  \geq 0 \nonumber \\
& \: \: \:\:\: \forall i \in \{1, \ldots, N\} \label{eq:CE-better1} 
\end{align} 
and
\begin{align} 
&\sum_{\bv{\omega} \in \Omega | \omega_i = v_i}\sum_{\bv{\alpha}\in\script{A}| \alpha_i=c_i} \pi[\bv{\omega}]Pr[\bv{\alpha}|\bv{\omega}]\theta_i(v_i,c_i) \nonumber \\
&\geq  \sum_{\bv{\omega} \in \Omega | \omega_i = v_i}\sum_{\bv{\alpha}\in\script{A}| \alpha_i=c_i}\pi[\bv{\omega}]Pr[\bv{\alpha}|\bv{\omega}] \hat{u}_i\left((\beta_i, \bv{\alpha}_{\overline{i}}), \bv{\omega}\right)  \nonumber \\
 & \: \: \: \: \forall i \in \{1, \ldots, N\} , \forall v_i \in \Omega_i, \forall c_i \in \script{A}_i, \forall \beta_i \in \script{A}_i \label{eq:CE-better2} 
 \end{align}
 
\item (CCE for the stochastic game) $Pr[\bv{\alpha}|\bv{\omega}]$ is a CCE for the stochastic game if there are real numbers 
 $\theta_i(v_i) \in [0, u_i^{max}]$ such that the constraints 
 \eqref{eq:NE-better1}-\eqref{eq:NE-better2} are satisfied. 
\end{itemize} 

The next lemma shows that every information restricted equilibrium $Pr[\bv{s}]$ 
 generates a general equilibrium $Pr[\bv{\alpha}|\bv{s}]$. 

\begin{lem} \label{lem:equiv-virtual} Suppose $Pr[\bv{s}]$ and $Pr[\bv{\alpha}|\bv{\omega}]$ satisfy \eqref{eq:generate}.  Then: 

(a)  $Pr[\bv{s}]$ satisfies the 
constraints \eqref{eq:NE-virtual}  if and only if $Pr[\bv{\alpha}|\bv{\omega}]$ satisfies the constraints
\eqref{eq:NE-better1}-\eqref{eq:NE-better2}. 

(b) $Pr[\bv{s}]$ satisfies the constraints \eqref{eq:CE-virtual}  if and only if $Pr[\bv{\alpha}|\bv{\omega}]$ satisfies the constraints \eqref{eq:CE-better1}-\eqref{eq:CE-better2}. 

(c) If $Pr[\bv{s}]$ is a NE for the virtual static game, then $Pr[\bv{\alpha}|\bv{\omega}]$ is a NE for the stochastic game. 
\end{lem} 

\begin{proof} 
See Appendix A. 
\end{proof}

One may wonder if the constraints \eqref{eq:NE-better1}-\eqref{eq:NE-better2} can be stated more simply by removing the 
$\theta_i(v_i)$ variables.  Indeed, one may wonder if \eqref{eq:NE-better1}-\eqref{eq:NE-better2} are equivalent to: 
\begin{align*} 
&\sum_{\bv{\omega} \in \Omega|\omega_i=v_i}\sum_{\bv{\alpha} \in \script{A}} \pi[\bv{\omega}]Pr[\bv{\alpha}|\bv{\omega}]\hat{u}_i(\bv{\alpha}, \bv{\omega})  \\
&\geq \sum_{\bv{\omega} \in \Omega|\omega_i=v_i}\sum_{\bv{\alpha} \in \script{A}} \pi[\bv{\omega}]Pr[\bv{\alpha}|\bv{\omega}]\hat{u}_i\left((\beta_i, \bv{\alpha}_{\overline{i}}), \bv{\omega}\right) 
\end{align*}
for all $i\in \{1, \ldots, N\}$, $v_i \in \Omega_i$, $\beta_i \in \script{A}_i$.  This is not generally the case. Indeed, the above constraints are more restrictive and imply that the conditional expected utility of player $i$, given she observes $\omega_i(t) = v_i$,  
is greater than or equal to the conditional expectation this player could achieve given $\omega_i(t) = v_i$ and given that she does not participate.  On the other hand, the constraints \eqref{eq:NE-better1}-\eqref{eq:NE-better2} allow a violation of this property for a given $v_i$. Such a violation does not imply that player $i$ could improve beyond the utility associated with participating. 
That is because that act of not participating may itself decrease the achievable average utility in certain $\omega_i$ states by an amount that cannot be recovered by changing strategies on other $\omega_i$ states.  This is a subtlety that does not arise in the static game context without the $\omega_i(t)$ processes.

Define $\script{E}_{NE}^{stoc}$, $\script{E}_{CE}^{stoc}$, $\script{E}_{CCE}^{stoc}$ as the set of all conditional probability 
mass functions $Pr[\bv{\alpha}|\bv{\omega}]$ that are NE, CE, and CCE, respectively, for the stochastic game. 

\begin{lem} For a general stochastic game as defined above: 

(a) The set $\script{E}_{NE}^{stoc}$ is nonempty. 

(b) $\script{E}_{NE}^{stoc} \subseteq \script{E}_{CE}^{stoc} \subseteq \script{E}_{CCE}^{stoc}$. 

(c) Sets $\script{E}_{CE}^{stoc}$ and $\script{E}_{CCE}^{stoc}$ are closed, bounded, and convex. 
\end{lem} 
 
 \begin{proof} 
 The virtual static game is finite and hence has at least one  mixed strategy NE $Pr[\bv{s}]$ \cite{nash-n-person-games}\cite{nash-games}.
 Let $Pr[\bv{\alpha}|\bv{\omega}]$ be the corresponding conditional mass function defined by \eqref{eq:generate}.  Then 
 $Pr[\bv{\alpha}|\bv{\omega}]$ is a NE for the stochastic game (by  Lemma \ref{lem:equiv-virtual}c), and so $\script{E}_{NE}^{stoc}$ is nonempty.  This proves part (a). 
 
 To prove (c), note that $\script{E}_{CCE}^{stoc}$ is the intersection of the set of all  $Pr[\bv{\alpha}|\bv{\omega}]$ 
 that satisfy  
 the (closed, bounded, and convex) probability simplex 
 constraints \eqref{eq:dist1}-\eqref{eq:dist2} and the set of all $Pr[\bv{\alpha}|\bv{\omega}]$ that satisfy the 
 linear constraints \eqref{eq:NE-better1}-\eqref{eq:NE-better2}.   Similarly, $\script{E}_{CE}^{stoc}$ is the intersection of all $Pr[\bv{\alpha}|\bv{\omega}]$ that satisfy \eqref{eq:dist1}-\eqref{eq:dist2} with all $Pr[\bv{\alpha}|\bv{\omega}]$ that satisfy the linear constraints \eqref{eq:CE-better1}-\eqref{eq:CE-better2}. 
 
 To prove that $\script{E}_{CE}^{stoc} \subseteq \script{E}_{CCE}^{stoc}$, suppose that $Pr[\bv{\alpha}|\bv{\omega}]$ is a CE. 
 Then it satisfies the CE constraints \eqref{eq:CE-better1}-\eqref{eq:CE-better2} for some values $\theta_i(v_i,c_i)$.  Define:\footnote{More precisely, the values $\theta_i(v_i)$ are defined to be $0$ in the special case when $\sum_{\bv{\omega}\in\Omega|\omega_i=v_i}\sum_{\bv{\alpha}\in\script{A}}\pi[\bv{\omega}]Pr[\bv{\alpha}|\bv{\omega}]= 0$.} 
 \begin{eqnarray*}
 \theta_i(v_i) \defequiv \frac{\sum_{\bv{\omega}\in\Omega | \omega_i=v_i}\sum_{\bv{\alpha}\in\script{A}}Pr[\bv{\omega}]Pr[\bv{\alpha}|\bv{\omega}]\theta_i(v_i, \alpha_i)}{\sum_{\bv{\omega}\in\Omega | \omega_i=v_i}\sum_{\bv{\alpha}\in\script{A}}Pr[\bv{\omega}]Pr[\bv{\alpha}|\bv{\omega}]} 
 \end{eqnarray*}
 These values satisfy $\theta_i(v_i) \in [0, u_i^{max}]$. Summing \eqref{eq:CE-better2} over $c_i \in \script{A}_i$ and applying the above definition of $\theta_i(v_i)$ proves that \eqref{eq:NE-better1}-\eqref{eq:NE-better2} hold.  The proof that $\script{E}_{NE}^{stoc} \subseteq \script{E}_{CE}^{stoc}$ is given in Appendix B. 
 \end{proof} 
 
 \subsection{Complexity comparison} 
 
 The CCE for the virtual static game is defined by the constraints \eqref{eq:NE-virtual}.  There is one such constraint for each $i \in \{1, \ldots, N\}$ and each $r_i \in \script{S}_i$, where $\script{S}_i$ is the number of pure strategies for player $i$. Thus,
 the number of constraints is: 
 \[ \sum_{i=1}^N|\script{S}_i| = \sum_{i=1}^N |\script{A}_i|^{|\Omega_i|} \]
This grows exponentially in the size of the sets $\Omega_i$. 
Thus, even though these constraints are linear, computation of a CCE for the virtual static game can be very complex.

The CCE constraints for the stochastic game are given in \eqref{eq:NE-better1}-\eqref{eq:NE-better2}.  There are $N$ constraints in \eqref{eq:NE-better1}.  For \eqref{eq:NE-better2}, there is one such constraint for each $i \in \{1, \ldots, N\}$, each $v_i \in \Omega_i$, and each $\beta_i \in \script{A}_i$, for a total  of: 
\[ N + \sum_{i=1}^N |\Omega_i||\script{A}_i| \]
This is \emph{linear} in the sizes of the $\Omega_i$ and $\script{A}_i$ sets. 
Thus, the general CCE constraints  \eqref{eq:NE-better1}-\eqref{eq:NE-better2}
provide a significant complexity reduction.    A similar ``exponential-to-polynomial'' complexity reduction holds for the CE definition 
 when  comparing  the constraints in \eqref{eq:CE-virtual} to those in  \eqref{eq:CE-better1}-\eqref{eq:CE-better2}.

\subsection{Unilateral changes cannot increase utility} 

The stochastic 
NE, CE, and CCE definitions above have the following property:  Assuming actions are chosen according to an equilibrium mass function $Pr[\bv{\alpha}|\bv{\omega}]$, a given player cannot improve her utility by unilaterally deviating from these actions.  This is formalized in the lemmas below.  

First note that if 
$\bv{\alpha}(t)$ is chosen according to a mass function $Pr[\bv{\alpha}|\bv{\omega}]$, then for all $i \in \{1, \ldots, N\}$: 
\[ \expect{u_i(t)} = \sum_{\bv{\omega}\in\Omega}\sum_{\bv{\alpha}\in\script{A}} \pi[\bv{\omega}]Pr[\bv{\alpha}|\bv{\omega}]\hat{u}_i(\bv{\alpha}, \bv{\omega}) \]

Now fix $i \in \{1, \ldots, N\}$.  For all $\omega_i \in \Omega$, let $X_i(\omega_i)$ be a random function that maps a point $\omega_i \in \Omega_i$ to a randomly chosen point $X_i(\omega_i) \in \script{A}_i$ according to some distribution that depends on $\omega_i$. 
It is assumed that for a given slot $t$,  $X_i(\omega_i(t))$ is conditionally independent of $\bv{\alpha}(t)$ and $\bv{\omega}(t)$ given $\omega_i(t)$.   The expected utility on slot $t$ associated with unilaterally changing action $\alpha_i(t)$ to action $X_i(\omega_i(t))$ is: 
\[ \sum_{\bv{\omega}\in\Omega}\sum_{\bv{\alpha}\in\script{A}} \pi[\bv{\omega}]Pr[\bv{\alpha}|\bv{\omega}]\expect{\hat{u}_i\left((X_i(\omega_i), \bv{\alpha}_{\overline{i}}), \bv{\omega}\right)} \]
where the expectation on the right-hand-side is with respect to the distribution of $X_i(\omega_i)$.

\begin{lem} \label{lem:equiv-all1} $Pr[\bv{\alpha}|\bv{\omega}]$ satisfies \eqref{eq:NE-better1}-\eqref{eq:NE-better2} 
if and only if: 
\begin{align} 
\expect{u_i(t)}  \geq \sum_{\bv{\omega}\in\Omega}\sum_{\bv{\alpha}\in\script{A}} \pi[\bv{\omega}]Pr[\bv{\alpha}|\bv{\omega}]\expect{\hat{u}_i\left((X_i(\omega_i), \bv{\alpha}_{\overline{i}}), \bv{\omega}\right)} \label{eq:unilateral-CCE} 
\end{align} 
for all $i \in \{1, \ldots, N\}$ and all randomized functions $X_i(\omega_i)$.
\end{lem}

\begin{proof} 
Suppose $Pr[\bv{\alpha}|\bv{\omega}]$ satisfies \eqref{eq:NE-better1}-\eqref{eq:NE-better2}. Fix $i \in \{1, \ldots, N\}$. Fix 
a random function $X_i(\omega_i)$, and define: 
\[ q_i(\beta_i| v_i) = Pr[X_i(v_i) = \beta_i] \: \: \forall v_i \in \Omega_i, \beta_i \in \script{A}_i \]
Multiplying \eqref{eq:NE-better2} by $q_i(\beta_i|v_i)$ and summing over $\beta_i \in \script{A}_i$ gives: 
\begin{align*} 
&\sum_{\bv{\omega}\in\Omega | \omega_i = v_i} \sum_{\bv{\alpha} \in \script{A}} \pi[\bv{\omega}]Pr[\bv{\alpha}|\bv{\omega}]\theta_i(v_i) \\
&\geq \sum_{\bv{\omega} \in \Omega | \omega_i = v_i}\sum_{\bv{\alpha} \in \script{A}}\pi[\bv{\omega}]Pr[\bv{\alpha}|\bv{\omega}]\sum_{\beta_i\in\script{A}_i} q_i(\beta_i|v_i)
\hat{u}_i((\beta_i, \bv{\alpha}_{\overline{i}}), \bv{\omega}) \\
&=\sum_{\bv{\omega} \in \Omega | \omega_i = v_i}\sum_{\bv{\alpha} \in \script{A}}\pi[\bv{\omega}]Pr[\bv{\alpha}|\bv{\omega}]\expect{
\hat{u}_i((X_i(v_i), \bv{\alpha}_{\overline{i}}), \bv{\omega}) }  
\end{align*} 
Summing both sides over $v_i \in \Omega_i$ and using \eqref{eq:NE-better1} gives the expression \eqref{eq:unilateral-CCE}. 

Now suppose \eqref{eq:unilateral-CCE} holds for all $i \in \{1, \ldots, N\}$ and all randomized functions $X_i(\omega_i)$. 
Fix $i \in \{1, \ldots, N\}$.  For each  $v_i \in \Omega_i$ deterministically 
define $X_i(v_i)$ as the element $\beta_i^*$ that maximizes
the right-hand-side of \eqref{eq:NE-better2} over all $\beta_i \in \script{A}_i$.  Likewise, define $\theta_i(v_i)$ by: 
 \begin{eqnarray*}
 \theta_i(v_i) \defequiv \frac{\sum_{\bv{\omega}\in\Omega | \omega_i=v_i}\sum_{\bv{\alpha}\in\script{A}}Pr[\bv{\omega}]Pr[\bv{\alpha}|\bv{\omega}]\hat{u}_i((X_i(v_i), \bv{\alpha}_{\overline{i}}), \bv{\omega})}{\sum_{\bv{\omega}\in\Omega | \omega_i=v_i}\sum_{\bv{\alpha}\in\script{A}}Pr[\bv{\omega}]Pr[\bv{\alpha}|\bv{\omega}]} 
 \end{eqnarray*}
 assuming the denominator is nonzero (else, define $\theta_i(v_i)=0$).  Then \eqref{eq:NE-better2} holds by construction.  Further, 
 inequality \eqref{eq:NE-better1} holds because it is equivalent to \eqref{eq:unilateral-CCE} for the given $X_i(v_i)$ function.  
\end{proof} 

The next lemma extends the random function $X_i(\omega_i)$ to $X_i(\omega_i, \alpha_i)$, so that its distribution depends on both $\omega_i$ and $\alpha_i$: 

\begin{lem} \label{lem:equiv-all3} $Pr[\bv{\alpha}|\bv{\omega}]$ is a CE for the stochastic game if and only if:  
\[ \expect{u_i(t)} \geq 
 \sum_{\bv{\omega}\in\Omega}\sum_{\bv{\alpha}\in\script{A}} \pi[\bv{\omega}]Pr[\bv{\alpha}|\bv{\omega}]\expect{\hat{u}_i\left((X_i(\omega_i, \alpha_i), \bv{\alpha}_{\overline{i}}), \bv{\omega}\right)}  \]
for all $i \in \{1, \ldots, N\}$ and all randomized functions $X_i(\omega_i, \alpha_i)$. 
\end{lem} 

\begin{proof} 
The proof is similar to that of Lemma \ref{lem:equiv-all1} and is omitted for brevity. 
\end{proof}

\subsection{Optimization objective} 

As before, define $\phi(u_1, \ldots, u_N)$ as a continuous and concave function  that maps $\times_{i=1}^N[0, u_i^{max}]$ 
to the set of real numbers. 
The goal is to choose messages $\bv{M}(t) = \bv{\alpha}(t)$ according to a conditional probability mass function 
$Pr[\bv{\alpha}(t) | \bv{\omega}(t)]$ that  solves the problem below: 
\begin{eqnarray}
\mbox{Maximize:} &  \phi(\overline{u}_1, \ldots, \overline{u}_N) \label{eq:p1}   \\
\mbox{Subject to:} & \overline{u}_i = \sum_{\bv{\omega} \in \Omega} \sum_{\bv{\alpha} \in \script{A}} \pi[\bv{\omega}]Pr[\bv{\alpha}|\bv{\omega}]\hat{u}_i(\bv{\alpha},\bv{\omega}) \nonumber  \\
& \: \: \: \: \forall i \in \{1, \ldots, N\} \label{eq:p2} \\
& \mbox{CCE constraints \eqref{eq:NE-better1}-\eqref{eq:NE-better2} are satisfied} \label{eq:p3}  \\
& Pr[\bv{\alpha}|\bv{\omega}] \geq 0 \: \: \forall \bv{\alpha} \in \script{A}, \bv{\omega} \in \Omega \label{eq:p4}  \\
& \sum_{\bv{\alpha}\in\script{A}} Pr[\bv{\alpha}|\bv{\omega}] = 1 \: \: \forall \bv{\omega} \in \Omega \label{eq:p5} 
\end{eqnarray} 

This is a convex program in the unknowns $Pr[\bv{\alpha}|\bv{\omega}]$.  The next section
presents an online solution that does not require knowledge of the probabilities $\pi[\bv{\omega}]$. 

\section{Lyapunov optimization} \label{section:lyap-opt}

For a real-valued stochastic process $u(t)$ defined over slots $t \in \{0, 1, 2, \ldots\}$, define: 
\[ \mbox{$\overline{u}(t) \defequiv \frac{1}{t}\sum_{\tau=0}^{t-1} \expect{u(\tau)}$} \]
Recall that $u_i(t) = \hat{u}_i(\bv{\alpha}(t), \bv{\omega}(t))$ for each player $i$ and each slot $t$. 
For $i \in \{1, \ldots, N\}$, $v \in \Omega_i$, and $\beta \in \script{A}_i$, define: 
\begin{eqnarray*}
u_{i,v}^{(\beta)}(t) &\defequiv& \hat{u}_i\left( (\beta, \bv{\alpha}_{\overline{i}}(t)), \bv{\omega}(t)\right)1\{\omega_i(t) = v\} 
\end{eqnarray*}
where $1\{\omega_i(t) = v\}$ is an indicator function that is 1 if $\omega_i(t) = v$, and 0 else.  The value
$u_{i,v}^{(\beta)}(t)$ is zero if $\omega_i(t) \neq v$, and else it is the utility player $i$ would receive on slot $t$ 
if it uses action  $\beta$ (assuming $\bv{\alpha}_{\overline{i}}(t)$ are the actions of others).

A reformulation of \eqref{eq:p1}-\eqref{eq:p5} that does not require the decisions $\bv{\alpha}(t)$ to use the same conditional distribution $Pr[\bv{\alpha}|\bv{\omega}]$ every slot is as follows: 
Every slot $t$, the game manager 
observes $\bv{\omega}(t)$ and chooses an action vector $\bv{\alpha}(t) \in \script{A}$ and variables $\theta_{i,v}(t)$ to solve: 
\begin{align}
&\hspace{-.1in}\mbox{Maximize:} \nonumber \\
& \hspace{+0in} \liminf_{t\rightarrow\infty} \phi(\overline{u}_1(t), \ldots, \overline{u}_N(t)) \label{eq:LL1-lowc} \\
&\hspace{-.1in}\mbox{Subject to:} \nonumber \\
&\hspace{+0in} \liminf_{t\rightarrow\infty} \left[\overline{u}_i(t) - \sum_{v\in\Omega_i}\overline{\theta}_{i,v}(t)\right] \geq 0 \nonumber \\
&\hspace{+1.6in}  \forall i \in \{1, \ldots, N\}  \label{eq:LLz-lowc} \\
&\hspace{+0in} \liminf_{t\rightarrow\infty} [\overline{\theta}_{i,v}(t) - \overline{u}_{i,v}^{(\beta)}(t)] \geq 0 \nonumber \\
&  \hspace{+.5in} \forall i \in \{1, \ldots, N\}, \forall v \in \Omega_i, \forall \beta \in \script{A}_i  \label{eq:LL2-lowc} \\
&\hspace{+0in} \bv{\alpha}(t) \in \script{A} \: \: \: \:  \forall t \label{eq:LLGGG} \\
&\hspace{+0in} \theta_{i,v}(t) \in [0, u_i^{max}1\{\omega_i(t)=v\}] \nonumber \\
& \hspace{+.9in} \forall t, \forall i \in \{1, \ldots, N\}, \forall v \in \Omega_i \label{eq:LL4-lowc} 
\end{align}

The constraints \eqref{eq:LLz-lowc} correspond to \eqref{eq:NE-better1}, and the constraints \eqref{eq:LL2-lowc} correspond
to \eqref{eq:NE-better2}.  Such time average problems can be solved by stationary and randomized algorithms \cite{sno-text}. Specifically, if $Pr[\bv{\alpha}|\bv{\omega}]$ and $\theta_i(v)$ are optimal variables for problem 
\eqref{eq:p1}-\eqref{eq:p5}, then the following is an optimal solution to \eqref{eq:LL1-lowc}-\eqref{eq:LL4-lowc}:  Every slot $t$, observe $\bv{\omega}(t)$ and independently 
choose $\bv{\alpha}(t)$ according to the conditional mass function 
$Pr[\bv{\alpha}|\bv{\omega}]$, and choose $\theta_{i,v}(t) = \theta_{i}(v)1\{\omega_i(t)=v\}$.  Conversely,  any solution to \eqref{eq:LL1-lowc}-\eqref{eq:LL4-lowc} has the following property:  For any $\epsilon>0$, there is a positive integer 
$T_{\epsilon}$ such that for any $t>T_{\epsilon}$,  time average expectations over $\{0, 1, \ldots, t-1\}$ produce 
conditional probability mass functions $Pr[\bv{\alpha}|\bv{\omega}]$ that are within $\epsilon$ of satisfying all constraints and achieving the optimal objective function value 
of problem  \eqref{eq:p1}-\eqref{eq:p5}.  Specifically: 
\begin{eqnarray*}
Pr[\bv{\alpha}|\bv{\omega}] = \left\{\begin{array}{cc}
& \frac{\frac{1}{t}\sum_{\tau=0}^{t-1} \expect{1\{\bv{\alpha}(t)=\bv{\alpha}, \bv{\omega}(t)=\bv{\omega}\}}}{\pi[\bv{\omega}]} , \mbox{ if $\pi[\bv{\omega}]>0$}\\
& 0 , \mbox{ if $\pi[\bv{\omega}] = 0$} 
\end{array}\right. 
\end{eqnarray*}

\subsection{Transformation via Jensen's inequality}

Using the auxiliary variable technique of \cite{sno-text}, the problem \eqref{eq:LL1-lowc}-\eqref{eq:LL4-lowc}, which seeks to maximize a nonlinear function of  a time average, can be transformed into a maximization of the time average of a nonlinear function. To this end, let 
 $\bv{\gamma}(t) = (\gamma_1(t), \ldots, \gamma_N(t))$ be an \emph{auxiliary vector} that the game manager
chooses on slot $t$, assumed to satisfy $0 \leq \gamma_i(t) \leq u_i^{max}$ for all $t$ and all $i$. 
Define: 
\[ g(t) \defequiv \phi(\gamma_1(t), \ldots, \gamma_N(t)) \]
Jensen's inequality implies that for all slots $t>0$: 
\begin{equation} \label{eq:jensen} 
 \overline{g}(t) \leq \phi(\overline{\gamma}_1(t), \ldots, \overline{\gamma}_N(t)) 
 \end{equation} 

Now consider the following problem:  Every slot $t \in \{0, 1, 2, \ldots\}$ the game manager 
observes $\bv{\omega}(t)$ and chooses an action vector $\bv{\alpha}(t) \in \script{A}$, variables $\theta_{i,v}(t)$,  and an auxiliary vector
$\bv{\gamma}(t)$ to solve:

\begin{align}
&\hspace{-.1in}\mbox{Maximize:} \nonumber \\
& \hspace{+0in} \liminf_{t\rightarrow\infty} \overline{g}(t) \label{eq:LL1-lowcmod} \\
&\hspace{-.1in}\mbox{Subject to:} \nonumber \\
& \lim_{t\rightarrow\infty} |\overline{\gamma}_i(t) - \overline{u}_i(t)| = 0 \: \: \forall i \in \{1, \ldots, N\} \label{eq:LLa-lowcmod} \\
&\hspace{+0in} \liminf_{t\rightarrow\infty} \left[\overline{u}_i(t) - \sum_{v\in\Omega_i}\overline{\theta}_{i,v}(t)\right] \geq 0 \nonumber \\
&\hspace{+1.7in}  \forall i \in \{1, \ldots, N\}  \label{eq:LLz-lowcmod} \\
&\hspace{+0in} \liminf_{t\rightarrow\infty} [\overline{\theta}_{i,v}(t) - \overline{u}_{i,v}^{(\beta)}(t)] \geq 0 \nonumber \\
&  \hspace{+.6in} \forall i \in \{1, \ldots, N\}, \forall v \in \Omega_i, \forall \beta \in \script{A}_i  \label{eq:LL2-lowcmod} \\
&\hspace{+0in} \bv{\alpha}(t) \in \script{A} \: \: \: \: \forall t \label{eq:GHGHGH} \\
&\hspace{+0in}   \theta_{i,v}(t) \in [0, u_i^{max}1\{\omega_i(t)=v\}] \nonumber \\
& \hspace{+1in} \forall t, \forall i \in \{1, \ldots, N\}, \forall v \in \Omega_i\label{eq:LL4-lowcmod} \\
& 0 \leq \gamma_i(t) \leq u_i^{max} \hspace{+.5in} \forall t, \forall i \in \{1, \ldots, N\} \label{eq:LL5-lowcmod} 
\end{align}

The problems \eqref{eq:LL1-lowc}-\eqref{eq:LL4-lowc} and \eqref{eq:LL1-lowcmod}-\eqref{eq:LL5-lowcmod}  are equivalent.  To see this, let $\phi_1^*$ and $\phi_2^*$ be the optimal objective values for problems 
\eqref{eq:LL1-lowc}-\eqref{eq:LL4-lowc} and \eqref{eq:LL1-lowcmod}-\eqref{eq:LL5-lowcmod}, respectively. 
Let $\bv{\alpha}^*(t)$ and $\theta_{i,v}^*(t)$ be optimal stationary and randomized decisions that solve
\eqref{eq:LL1-lowc}-\eqref{eq:LL4-lowc}, and let 
$u_i^* \defequiv \expect{\hat{u}_i(\bv{\alpha}^*(t), \bv{\omega}(t))}$ be the corresponding expected utilities for player $i$. Then: 
\[ \phi_1^* = \phi(u_1^*, \ldots, u_N^*) \]
The decisions $\bv{\alpha}^*(t)$ and $\theta_{i,v}^*(t)$ can be used, together with  
$\gamma_i(t)  = u_i^*$ for all $t$ and all $i \in \{1, \ldots, N\}$, to satisfy all constraints of the new 
problem \eqref{eq:LL1-lowcmod}-\eqref{eq:LL5-lowcmod} with (possibly sub-optimal) objective function value $\overline{g} = \phi_1^*$. Because this is not necessarily optimal for the new problem, one has $\phi_2^* \geq \phi_1^*$. 

On the other hand, let $\bv{\alpha}(t)$, $\theta_{i,v}(t)$, and $\gamma_i(t)$ be decisions that solve the new problem 
\eqref{eq:LL1-lowcmod}-\eqref{eq:LL5-lowcmod}.  Then these same decisions satisfy all constraints of the problem 
\eqref{eq:LL1-lowc}-\eqref{eq:LL4-lowc} and thus yield an objective function value no more than $\phi_1^*$, so that: 
\begin{eqnarray}
\phi_1^* &\geq& \liminf_{t\rightarrow\infty} \phi(\overline{u}_1(t), \ldots, \overline{u}_N(t)) \nonumber\\
&=&  \liminf_{t\rightarrow\infty} \phi(\overline{\gamma}_1(t), \ldots, \overline{\gamma}_N(t)) \label{eq:equiv1} \\
&\geq& \liminf_{t\rightarrow\infty} \overline{g}(t) \label{eq:equiv2} \\
&=& \phi_2^* \label{eq:equiv3}
\end{eqnarray}
where \eqref{eq:equiv1} follows by \eqref{eq:LLa-lowcmod} together with 
continuity of $\phi(\cdot)$, \eqref{eq:equiv2} follows by Jensen's inequality \eqref{eq:jensen}, and 
\eqref{eq:equiv3} follows because the decisions are optimal for the new problem. 
It follows that $\phi_1^* = \phi_2^*$.  In particular, 
any solution to \eqref{eq:LL1-lowcmod}-\eqref{eq:LL5-lowcmod} also solves \eqref{eq:LL1-lowc}-\eqref{eq:LL4-lowc}.

\subsection{The drift-plus-penalty algorithm} \label{section:alg}

For the constraints \eqref{eq:LLz-lowcmod}, for each $i \in \{1, \ldots, N\}$ 
define a \emph{virtual queue} $Q_{i}(t)$ with update equation: 
\begin{equation*}
Q_{i}(t+1) = \max\left[Q_{i}(t) + \sum_{v\in\Omega_i}\theta_{i,v}(t) - u_i(t), 0\right] 
\end{equation*} 
The above looks like a slotted time queueing equation with arrival process $\sum_{v \in \Omega_i} \theta_{i,v}(t)$ and service process
$u_i(t)$. The intuition is that if a control algorithm is constructed that makes these queues \emph{mean rate stable}, so that: 
\[ \lim_{t\rightarrow\infty} \frac{\expect{Q_i(t)}}{t} = 0 \]
then constraint \eqref{eq:LLz-lowcmod} is satisfied \cite{sno-text}. This queueing update can be simplified using the identity: 
\[ \sum_{v\in\Omega_i} \theta_{i,v}(t) = \theta_{i,\omega_i(t)}(t) \]
Hence: 
\begin{equation} \label{eq:q-update} 
Q_{i}(t+1) = \max[Q_{i}(t) + \theta_{i,\omega_i(t)}(t) - u_i(t), 0] 
\end{equation} 

Likwewise, to enforce the constraint 
\eqref{eq:LL2-lowcmod}, for each $i\in\{1, \ldots, N\}$, $v\in\Omega_i$, $\beta \in \script{A}_i$, define 
a virtual queue $J_{i,v}^{(\beta)}(t)$ with update equation: 
\begin{equation} \label{eq:j-update} 
J_{i,v}^{(\beta)}(t+1) = \max[J_{i,v}^{(\beta)}(t) + u_{i,v}^{(\beta)}(t) - \theta_{i, v}(t), 0] 
\end{equation} 

Finally, for the constraints \eqref{eq:LLa-lowcmod}, for each $i \in \{1, \ldots, N\}$ define a virtual queue $Z_i(t)$ with update equation: 
\begin{equation} \label{eq:z-update} 
Z_i(t+1) = Z_i(t) + \gamma_i(t) - u_i(t) 
\end{equation} 

Define the function: 
\begin{eqnarray*}
 L(t) &=& \frac{1}{2} \sum_{i=1}^N Z_i(t)^2 + \frac{1}{2}\sum_{i=1}^NQ_i(t)^2 \\
 && +  \frac{1}{2}\sum_{i=1}^N\sum_{v\in\Omega_i, \beta\in\script{A}_i} J_{i,v}^{(\beta)}(t)^2 
 \end{eqnarray*}
 This is called a \emph{Lyapunov function}. 
 Define $\Delta(t) \defequiv L(t+1) - L(t)$, called the \emph{Lyapunov drift} on slot $t$. The 
 drift-plus-penalty algorithm is defined by choosing control actions greedily every slot to minimize
 a bound on the \emph{drift-plus-penalty expression} $\Delta(t) - Vg(t)$. Here, $-g(t)$ is the ``penalty'' and $V$
 is a nonnegative constant that affects a tradeoff between convergence time and proximity to the optimal solution. 
 
 \begin{lem} For all slots $t$ one has: 
 \begin{align}
 &\Delta(t) - Vg(t)   \leq \nonumber  \\
 & B - Vg(t) + \sum_{i=1}^N Z_i(t)(\gamma_i(t) - u_i(t)) \nonumber \\
 & + \sum_{i=1}^N Q_i(t)(\theta_{i,\omega_i(t)}(t) - u_i(t)) \nonumber \\
 & + \sum_{i=1}^N\sum_{v\in\Omega_i, \beta\in\script{A}_i}J_{i,v}^{(\beta)}(t)(u_{i,v}^{(\beta)}(t)-\theta_{i,v}(t)) \label{eq:dpp-lem} 
 \end{align} 
 where: 
 \[ B \defequiv \sum_{i=1}^N (u_i^{max})^2 + \frac{1}{2}\sum_{i=1}^N |\script{A}_i|(u_i^{max})^2 \]
 \end{lem} 
 
 \begin{proof} 
 From \eqref{eq:z-update} one has: 
 \begin{eqnarray*}
  \frac{Z_i(t+1)^2}{2} &=& \frac{Z_i(t)^2}{2}  + \frac{(\gamma_i(t) - u_i(t))^2}{2}  \\
  && + Z_i(t)(\gamma_i(t) - u_i(t)) 
  \end{eqnarray*}
 From \eqref{eq:q-update} and the fact that $\max[x, 0]^2 \leq x^2$, one has: 
 \begin{eqnarray*}
 \frac{Q_i(t+1)^2}{2} &\leq& \frac{Q_i(t)^2}{2} + \frac{(\theta_{i,\omega_i(t)}(t)-u_i(t))^2}{2} \\
 && + Q_i(t)(\theta_{i,\omega_i(t)}(t) - u_i(t)) 
 \end{eqnarray*}
 Similarly, from \eqref{eq:j-update}: 
 \begin{eqnarray*}
 \frac{J_{i,v}^{(\beta)}(t+1)^2}{2} &\leq& \frac{J_{i,v}^{(\beta)}(t)^2}{2} + \frac{(u_{i,v}^{(\beta)}(t) - \theta_{i,v}(t))^2}{2} \\
 && + J_{i,v}^{(\beta)}(t)(u_{i,v}^{(\beta)}(t) - \theta_{i,v}(t))
 \end{eqnarray*}
 Summing the above yields: 
 \begin{eqnarray*}
 \Delta(t) &\leq& B(t) + \sum_{i=1}^N Z_i(t)(\gamma_i(t) - u_i(t)) \\
 && + \sum_{i=1}^N Q_i(t)(\theta_{i,\omega_i(t)}(t) - u_i(t))  \\
 && +  \sum_{i=1}^N\sum_{v\in\Omega_i, \beta\in\script{A}_i}J_{i,v}^{(\beta)}(t)(u_{i,v}^{(\beta)}(t)-\theta_{i,v}(t)) 
 \end{eqnarray*}
 where $B(t)$ is a value that satisfies $B(t) \leq B$ for all $t$.  Adding $-Vg(t)$ to both sides proves the result. 
 \end{proof} 
 
Greedily minimizing the right-hand-side of \eqref{eq:dpp-lem} every slot $t$ leads to the following algorithm: 
Every slot $t$,  the game manager observes the queues $Z_i(t)$, $Q_i(t)$, $J_{i,v}^{(\beta)}(t)$ and the current $\bv{\omega}(t)$.  Then: 

\begin{itemize} 
\item \emph{Auxiliary variables $\gamma_i(t)$:} The game manager chooses $\bv{\gamma}(t) = (\gamma_1(t), \ldots, \gamma_N(t))$ 
as the solution to: 
\begin{eqnarray} 
\mbox{Maximize:} & V\phi(\gamma_1(t), \ldots, \gamma_N(t)) - \sum_{i=1}^NZ_i(t)\gamma_i(t) \nonumber \\
\mbox{Subject to:} & 0 \leq \gamma_i(t) \leq u_i^{max} \: \: \forall i \in \{1, \ldots, N\} \label{eq:gamma-solve} 
\end{eqnarray}

\item \emph{Auxiliary variables $\theta_{i,v}(t)$:} For each $i \in \{1, \ldots, N\}$ and 
$v\in\Omega_i$, choose $\theta_{i,v}(t) \in [0, u_i^{max}1\{\omega_i(t) = v\}]$ to minimize: 
\[ \sum_{i=1}^NQ_i(t)\theta_{i,\omega_i(t)}(t) - \sum_{i=1}^N\sum_{v\in\Omega_i, \beta \in \script{A}_i} J_{i,v}^{(\beta)}(t)\theta_{i,v}(t) \]

\item \emph{Suggested actions:}  Choose $\bv{\alpha}(t) \in \script{A}_1 \times \cdots \times \script{A}_N$ to minimize: 
\begin{eqnarray*}
-\sum_{i=1}^N[Z_i(t)+Q_i(t)]\hat{u}_i(\bv{\alpha}(t), \bv{\omega}(t)) \\
+ \sum_{i=1}^N\sum_{v\in\Omega_i, \beta\in\script{A}_i} J_{i,v}^{(\beta)}(t)\hat{u}_i\left( (\beta, \bv{\alpha}_{\overline{i}}(t)), \bv{\omega}(t)\right)1\{\omega_i(t) = v\} 
\end{eqnarray*}
The manager then sends suggested actions $\alpha_i(t)$ to each (participating) player
$i \in \{1, \ldots, N\}$. 

\item \emph{Queue update:} Update virtual queues via  \eqref{eq:q-update}, \eqref{eq:j-update}, \eqref{eq:z-update}. 
\end{itemize} 

This is an online algorithm that does not require knowledge of the probabilities $\pi[\bv{\omega}]$. 

\subsection{A closer look at the algorithm} 

The $\theta_{i,v}(t)$ selection in the above algorithm reduces to the following:  Every slot $t$, observe $\bv{\omega}(t)$ and the queues.  Then for each $i \in \{1, \ldots, N\}$ and $v \in \Omega_i$, choose: 
\[ \theta_{i,v}(t) = \left\{\begin{array}{ll}
u_i^{max} & \mbox{ if $\omega_i(t)=v$ and $Q_i(t) < \sum_{\beta\in\script{A}_i}J_{i,v}^{(\beta)}(t)$} \\
0 & \mbox{ otherwise}  
\end{array}\right. \]

The $\bv{\alpha}(t)$ decisions reduce to the following: Every slot $t$, observe $\bv{\omega}(t)$ and the queues. Then choose $\bv{\alpha}(t) \in \script{A}$ to minimize: 
\begin{eqnarray*}
 -\sum_{i=1}^N[Z_i(t) + Q_i(t)]\hat{u}_i(\bv{\alpha}(t), \bv{\omega}(t)) \\
 + \sum_{i=1}^N\sum_{\beta\in\script{A}_i} J_{i,\omega_i(t)}^{(\beta)}(t) \hat{u}_i\left((\beta, \bv{\alpha}_{\overline{i}}(t)), \bv{\omega}(t)  \right) 
\end{eqnarray*}

Finally, consider the case when the fairness function  is a separable sum of individual concave functions:  
\[ \phi(\gamma_1, \ldots, \gamma_N) = \sum_{i=1}^N \phi_i(\gamma_i) \]
Then the $\bv{\gamma}(t)$ decisions reduce to separately choosing $\gamma_i(t)$ for each $i \in \{1, \ldots, N\}$ as the value in the interval $[0, u_i^{max}]$ that maximizes $V\phi_i(\gamma_i(t)) - Z_i(t)\gamma_i(t)$.  For example, if 
$\phi_i(\gamma_i) = \log(1 + \gamma_i)$, then: 
\[ \gamma_i(t) = \left\{\begin{array}{cc}
u_i^{max} & \mbox{, if $Z_i(t) \leq 0$} \\
\left[ \frac{V}{Z_i(t)} - 1 \right]_{0}^{u_i^{max}} & \mbox{, if $Z_i(t)>0$}\end{array}\right. \]
where $[x]_0^{a}$ is defined: 
\[ [x]_0^a \defequiv \left\{\begin{array}{cc}
0 & \mbox{ if $x < 0$}\\
x & \mbox{ if $0 \leq x \leq a$} \\
a & \mbox{ if $x > a$} \end{array}\right. \]

\subsection{Performance analysis} 

For simplicity, assume all virtual queues are initially empty, so that $L(0)=0$. 
Define $\phi^*$ as the optimal value of the objective function for \eqref{eq:p1}-\eqref{eq:p5}. By 
equivalence of the transformations, $\phi^*$  is also the optimal value for problem \eqref{eq:LL1-lowcmod}-\eqref{eq:LL5-lowcmod}.  Define $\bv{X}(t)$ as the vector of all virtual queues $Z_i(t)$, $Q_i(t)$, $J_{i,v}^{(\beta)}(t)$, 
and define $\norm{\bv{X}(t)} \defequiv \sqrt{2L(t)}$. 

\begin{thm} \label{thm:lyap-performance} If $L(0)=0$ and the 
above algorithm is implemented using a fixed value $V\geq 0$, then: 

(a) For all slots $t>0$ one has: 
\[ \phi(\overline{\gamma}_1(t), \ldots, \overline{\gamma}_N(t)) \geq \phi^* - B/V \]

(b) All virtual queues $Z_i(t)$, $Q_i(t)$, $J_{i,v}^{(\beta)}(t)$ are mean rate stable, so all  constraints
\eqref{eq:LLa-lowcmod}-\eqref{eq:LL5-lowcmod} are satisfied. 

(c) For all slots $t>0$ the virtual queue sizes satisfy: 
\[ \frac{\expect{\norm{\bv{X}(t)}}}{t} \leq \sqrt{\frac{2B + 2V(g_{max} - \phi^*)}{t}} \]
where $g_{max}$ is the maximum possible value for $g(t)$, being the maximum of $\phi(\gamma_1, \ldots, \gamma_N)$ over $\gamma_i \in [0, u_i^{max}]$ for all $i \in \{1, \ldots, N\}$.\footnote{In the special case when $\phi(\gamma_1, \ldots, \gamma_N)$ is entrywise nondecreasing, then $g_{max} \defequiv \phi(u_1^{max}, \ldots, u_N^{max})$.}   
\end{thm} 

\begin{proof} 
Fix a time slot $t$.  Given the existing queue values $Z_i(t)$, $Q_i(t)$, $J_{i,v}^{(\beta)}(t)$ and the observed
$\bv{\omega}(t)$, the algorithm 
makes decisions $\gamma_i(t)$, $\theta_{i,v}(t)$, $\bv{\alpha}(t)$ to minimize 
the right-hand-side of \eqref{eq:dpp-lem}.  Thus: 
\begin{align}
&\Delta(t) - Vg(t) \leq \nonumber \\
&B - V\phi(\gamma_1^*(t), \ldots, \gamma_N^*(t)) \nonumber \\
& + \sum_{i=1}^NZ_i(t)(\gamma_i^*(t) - u_i^*(t)) \nonumber \\
&+\sum_{i=1}^NQ_i(t)(\theta^*_{i,\omega_i(t)}(t) - u_i^*(t)) \nonumber \\
&+\sum_{i=1}^N\sum_{v\in\Omega_i, \beta\in\script{A}_i}J_{i,v}^{(\beta)}(t)(u_{i,v}^{*(\beta)}(t) - \theta_{i,v}^*(t)) \label{eq:dpp-proof} 
\end{align}
for any alternative decisions $\gamma_i^*(t)$, $\theta_{i,v}^*(t)$, $\bv{\alpha}^*(t)$ that satisfy \eqref{eq:LL4-lowcmod}-\eqref{eq:LL5-lowcmod}, and where: 
\begin{eqnarray*}
u_i^*(t) &\defequiv& \hat{u}_i(\bv{\alpha}^*(t), \bv{\omega}(t)) \\
u_{i,v}^{*(\beta)}(t) &\defequiv& \hat{u}_i\left( \left( \beta , \bv{\alpha}_{\overline{i}}^*(t) \right)  , \bv{\omega}(t)  \right) 
\end{eqnarray*}
Now consider alternative decisions defined by the optimal solution to problem \eqref{eq:p1}-\eqref{eq:p5}. 
Specifically, choose $\bv{\alpha}^*(t) \in \script{A}$ to be conditionally independent of  current queue states, given the observed $\bv{\omega}(t)$, according to the probability distribution $Pr[\bv{\alpha}|\bv{\omega}]$ that solves
\eqref{eq:p1}-\eqref{eq:p5}.  Choose $\theta_{i,v}^*(t) = \theta_{i,v}^*$, where $\theta_{i,v}^*$ are the optimal values for the solution to \eqref{eq:p1}-\eqref{eq:p5}.  Finally, define $u_i^* = \expect{u_i^*(t)}$, being the expected utility of player $i$ under the optimal distribution $Pr[\bv{\alpha}|\bv{\omega}]$, and note that: 
\[ \phi(u_1^*, \ldots, u_N^*) = \phi^* \]
Choose $\gamma_i^*(t) = u_i^*$ for all $i \in \{1, \ldots, N\}$.  Then \eqref{eq:NE-better1}-\eqref{eq:NE-better2} 
imply: 
\begin{eqnarray}
\expect{\phi(\gamma_1^*(t), \ldots, \gamma_N^*(t))} &=& \phi^* \label{eq:ab1} \\
\expect{\gamma_i^*(t)} &=& \expect{u_i^*(t)}  \label{eq:ab2} \\
\expect{\theta_{i, \omega_i(t)}^*(t)} &\leq& \expect{u_i^*(t)}   \label{eq:ab3} \\
\expect{u_{i,v}^{*(\beta)}(t)} &\leq& \expect{\theta_{i,v}^*(t)} \label{eq:ab4} 
\end{eqnarray}
Taking expectations of \eqref{eq:dpp-proof} and substituting \eqref{eq:ab1}-\eqref{eq:ab4} gives: 
\[ \expect{\Delta(t) - Vg(t)} \leq B - V\phi^* \]
The above inequality holds for all $t \in \{0, 1, 2, \ldots\}$. Fix a slot $T>0$. Summing the above over slots $t \in \{0, 1, 2, \ldots, T-1\}$ and using $L(0) = 0$ gives: 
\begin{equation} \label{eq:reuse-proof} 
\expect{L(T)} - V\sum_{t=0}^{T-1} \expect{g(t)} \leq BT - V\phi^* T 
\end{equation} 
Rearranging \eqref{eq:reuse-proof} and using the definition of $g(t)$ gives: 
\[ \frac{1}{T}\sum_{t=0}^{T-1} \expect{\phi(\gamma_1(t), \ldots, \gamma_N(t))} \geq \phi^* - \frac{B}{V} - \frac{\expect{L(T)}}{VT} \]
Using Jensen's inequality and $\expect{L(T)} \geq 0$ proves part (a). 

Again rearranging \eqref{eq:reuse-proof} gives: 
\[ \expect{\norm{\bv{X}(T)}^2} \leq 2BT + 2VT (g_{max} - \phi^*) \]
Using the fact that $\expect{\norm{\bv{X}(T)}}^2 \leq \expect{\norm{\bv{X}(t)}^2}$, dividing by $T^2$, and taking square roots proves part (c).  Part (b) follows immediately from part (c). 

\end{proof} 

Define $\epsilon = 1/V$.  Theorem \ref{thm:lyap-performance} shows that average utility is within 
$O(\epsilon)$ of optimality.  Part (c) of the theorem implies that constraint violation is within $O(\epsilon)$ after
time $O(1/\epsilon^3)$.  If a \emph{Slater condition} holds, this convergence time is improved to $O(1/\epsilon^2)$ \cite{sno-text}. Similar bounds can be shown for infinite horizon time averages (rather than time average expectations) \cite{neely-lyap-opt}. 

\subsection{Discussion} 

The online algorithm ensures the constraints \eqref{eq:LLz-lowcmod}-\eqref{eq:LL2-lowcmod} are satisfied.  This shows that average utility of each player $i$ is greater than or equal to the achievable utility if the player were to constantly use the best pure strategy.  The best pure strategy of player $i$ is the one that uses the optimal action $\alpha_i^*(\omega_i)$ as a function of the observed $\omega_i(t)$.   This corresponds to the constraints in \eqref{eq:NE-better1}-\eqref{eq:NE-better2}.   If an algorithm makes random decisions independently every slot according to a conditional probability mass function $Pr[\bv{\alpha}|\bv{\omega}]$,  then constraints \eqref{eq:LLz-lowcmod}-\eqref{eq:LL2-lowcmod} imply player $i$ cannot do better under \emph{any} alternative decisions, possibly those that mix pure strategies with different mixing probabilities every slot.  A subtlety is that the online algorithm does \emph{not} make stationary and randomized decisions.  Thus, it is not clear if a player with knowledge of the algorithm could improve average utility by making alternative decisions that do not correspond to a pure strategy.   Of course, the online algorithm yields time averages that correspond to a desired $Pr[\bv{\alpha}|\bv{\omega}]$.  Thus, a potential fix is to run the online algorithm in the background and make $\bv{\alpha}(t)$ decisions according to the  time averages that emerge.   

\section{Simplification under a special case}

Consider the special case when there is a single random event process $\omega_0(t)$ that is known only to the game manager.  Thus, there are no random event processes $\omega_i(t)$ for any player $i \in \{1, \ldots, N\}$. 
This can be treated in the framework of the previous section by formally defining the sets $\Omega_i$ to consist of a single element 0, so that $\omega_i(t)=0$ for all slots $t$ and all players $i \in \{1, \ldots, N\}$.   However, this special case can be treated more simply by removing the auxiliary variables $\theta_{i,v}(t)$. 
Indeed, 
for all $i \in \{1, \ldots, N\}$ and all $\beta \in \script{A}_i$, define: 
\begin{eqnarray*}
u_i(t) &=& \hat{u}_i\left(\bv{\alpha}(t), \omega_0(t)\right) \\
 u_i^{(\beta)}(t) &=& \hat{u}_i\left((\beta, \bv{\alpha}_{\overline{i}}(t)), \omega_0(t)\right) 
\end{eqnarray*}

In this special case, the problem  \eqref{eq:LL1-lowcmod}-\eqref{eq:LL5-lowcmod} reduces to:

\begin{align}
&\hspace{-.1in}\mbox{Maximize:} \nonumber \\
& \hspace{+0in} \liminf_{t\rightarrow\infty} \overline{g}(t) \label{eq:sp1} \\
&\hspace{-.1in}\mbox{Subject to:} \nonumber \\
& \lim_{t\rightarrow\infty} |\overline{\gamma}_i(t) - \overline{u}_i(t)| = 0 \: \: \forall i \in \{1, \ldots, N\} \label{eq:sp2} \\
&\hspace{+0in} \liminf_{t\rightarrow\infty} \left[\overline{u}_i(t) - \overline{u}_i^{(\beta)}(t)\right] \geq 0 \nonumber \\
&  \hspace{+.6in} \forall i \in \{1, \ldots, N\},  \forall \beta \in \script{A}_i  \label{eq:sp3} \\
&\hspace{+0in} \bv{\alpha}(t) \in \script{A} \: \: \forall t, \forall i \in \{1, \ldots, N\} \label{eq:sp4} \\
& 0 \leq \gamma_i(t) \leq u_i^{max} \hspace{+.5in} \forall t, \forall i \in \{1, \ldots, N\} \label{eq:sp5} 
\end{align}
The above constraints are different from \eqref{eq:LL1-lowcmod}-\eqref{eq:LL5-lowcmod} because the  variables $\theta_{i,v}(t)$ have been removed, the constraint \eqref{eq:LL2-lowcmod} has been removed, and  the constraint \eqref{eq:LLz-lowcmod} has been modified to \eqref{eq:sp3}.  

Since the constraint \eqref{eq:sp2} is identical to constraint \eqref{eq:LLa-lowcmod}, it is enforced by
the same virtual queue $Z_i(t)$ with update equation given in  \eqref{eq:z-update}.  However, the constraint 
\eqref{eq:sp3} is enforced by the following new constraint for all $i \in \{1, \ldots, N\}$ and $\beta \in \script{A}_i$: 
\begin{equation} \label{eq:queue-new} 
Q_i^{(\beta)}(t+1) =   \max[Q_i^{(\beta)}(t) + u_i^{(\beta)}(t) - u_i(t), 0]  
\end{equation}

The resulting algorithm is as follows: Every slot $t$, the game manager observes the queues $Z_i(t), 
Q_i^{(\beta)}(t)$ and the current $\omega_0(t)$.  Then: 
\begin{itemize} 
\item \emph{Auxiliary variables $\gamma_i(t)$:}  Choose
$(\gamma_1(t), \ldots, \gamma_N(t))$ as before (that is, according to \eqref{eq:gamma-solve}). 

\item \emph{Suggested actions:} Choose $\bv{\alpha}(t) \in \script{A}_i\times \cdots \times \script{A}_N$
to minimize: 
\begin{align*} 
& -\sum_{i=1}^N \left[Z_i(t)+\sum_{\beta\in\script{A}_i}Q_i^{(\beta)}(t)\right]\hat{u}_i(\bv{\alpha}(t), \omega_0(t)) \\
& + \sum_{i=1}^N\sum_{\beta\in\script{A}_i} Q_i^{(\beta)}(t)\hat{u}_i((\beta, \bv{\alpha}_{\overline{i}}(t)), \omega_0(t)) 
\end{align*} 

\item \emph{Virtual queue udpate:} Update queues $Z_i(t)$ and $Q_i^{(\beta)}(t)$ via
\eqref{eq:z-update} and 
\eqref{eq:queue-new}. 
\end{itemize} 

In this special case when no player observes any random events, the set of pure strategies for each player $i$ coincides with the set of actions $\script{A}_i$. Thus, this 
algorithm is the same as that given in the conference version of this paper \cite{repeated-games-allerton}, where it is shown to give  performance similar to that of Theorem 
\ref{thm:lyap-performance}.

%
%

\section{Conclusions} 

This paper considered a simple game structure for repeated stochastic games.  Every slot a 
random vector $\bv{\omega}(t)$ is generated by nature.  Players observe different components 
of this vector and then choose individual actions with the help of a game manager. 
A coarse correlated equilibrium (CCE) for this stochastic game was defined to ensure that participating players earn at least as much utility as they could earn by individually deviating (and hence receiving no help from the manager).   The paper considered optimizing a concave function of the vector of time average utilities (called a fairness function), subject to the CCE constraints.  Lyapunov optimization was used to solve the problem over time, without requiring knowledge of the probabilities for the $\bv{\omega}(t)$ process. Similar techniques can be used to enforce correlated equilibrium (CE) constraints.

\section*{Appendix A --- Proof of Lemma \ref{lem:equiv-virtual}} 

Suppose $Pr[\bv{s}]$ and $Pr[\bv{\alpha}|\bv{\omega}]$ satisfy \eqref{eq:generate}. 
The following identities are useful. For all $i \in \{1, \ldots, N\}$ one has: 
\begin{equation}
\sum_{\bv{\omega}\in\Omega}\sum_{\bv{\alpha}\in\script{A}} \pi[\bv{\omega}]Pr[\bv{\alpha}|\bv{\omega}]\hat{u}_i(\bv{\alpha}, \bv{\omega}) = \sum_{\bv{s}\in\script{S}} Pr[\bv{s}]h_i(\bv{s}) \label{eq:useful-id1} 
\end{equation} 
This can be proven by substituting \eqref{eq:generate} into the left-hand-side of \eqref{eq:useful-id1} and using \eqref{eq:hi}. Likewise, for any $i \in \{1, \ldots, N\}$ and any pure strategy function $b_i^{(r_i)}(\omega_i)$ for player 
$i$ (where $r_i \in \script{S}_i$), one has: 
\begin{align} 
&\sum_{\bv{\omega}\in\Omega}\sum_{\bv{\alpha}\in\script{A}} \pi[\bv{\omega}]Pr[\bv{\alpha}|\bv{\omega}]\hat{u}_i((b_i^{(r_i)}(\omega_i), \bv{\alpha}_{\overline{i}}), \bv{\omega})  \nonumber \\
&= \sum_{\bv{s}\in\script{S}} Pr[\bv{s}]h_i(r_i, \bv{s}_{\overline{i}}) \label{eq:useful-id2} 
\end{align}

\begin{proof} (Lemma \ref{lem:equiv-virtual}a)  
Suppose $Pr[\bv{\alpha}|\bv{\omega}]$ satisfies the constraints \eqref{eq:NE-better1}-\eqref{eq:NE-better2}.  
Fix $i \in \{1, \ldots, N\}$. Substituting \eqref{eq:useful-id1} into the left-hand-side of \eqref{eq:NE-better1} gives: 
\begin{align} 
\sum_{\bv{s}\in\script{S}}Pr[\bv{s}]h_i(\bv{s}) \geq \sum_{\bv{\omega}\in\Omega}\sum_{\bv{\alpha}\in\script{A}}\pi[\bv{\omega}]Pr[\bv{\alpha}|\bv{\omega}]\theta_i(\omega_i) \label{eq:yuay1} 
\end{align} 
Now fix an index $r_i \in \script{S}_i$.  For each $v_i \in \Omega_i$, define $\beta_i = b_i^{(r_i)}(v_i)$.  Substituting this $\beta_i$ into \eqref{eq:NE-better2} gives: 
\begin{align*} 
 &\sum_{\bv{\omega} \in \Omega | \omega_i = v_i}\sum_{\bv{\alpha}\in\script{A}} \pi[\bv{\omega}] Pr[\bv{\alpha}|\bv{\omega}] 
\theta_i(v_i) \nonumber \\
& \geq \sum_{\bv{\omega}\in\Omega| \omega_i=v_i} \sum_{\bv{\alpha}\in\script{A}} \pi[\bv{\omega}]Pr[\bv{\alpha}|\bv{\omega}]   \hat{u}_i((b_i^{(r_i)}(v_i), \bv{\alpha}_{\overline{i}}), \bv{\omega}) 
\end{align*}
Summing the above over all $v_i \in \Omega_i$ and using \eqref{eq:useful-id2} gives: 
\begin{align*} 
 &\sum_{\bv{\omega} \in \Omega}\sum_{\bv{\alpha}\in\script{A}} \pi[\bv{\omega}] Pr[\bv{\alpha}|\bv{\omega}] 
\theta_i(v_i)  \geq \sum_{\bv{s}\in\script{S}}Pr[\bv{s}] h_i(r_i, \bv{s}_{\overline{i}})
\end{align*}
Combining the above with \eqref{eq:yuay1} 
proves that  the constraints \eqref{eq:NE-virtual} hold. 

Now Suppose $Pr[\bv{s}]$ satisfies the constraints \eqref{eq:NE-virtual}.  
For each $i \in \{1, \ldots, N\}$ and $v_i \in \Omega_i$, define the function $b_i(v_i)$ as the action $\beta_i \in \script{A}_i$ 
that maximizes: 
\[ \sum_{\bv{\omega}\in\Omega| \omega_i=v_i} \sum_{\bv{\alpha} \in \script{A}} \pi[\bv{\omega}]Pr[\bv{\alpha}|\bv{\omega}]\hat{u}_i\left((\beta_i, \bv{\alpha}_{\overline{i}}), \bv{\omega} \right) \]
Define $\theta_i(v_i)$ as the corresponding maximum divided by $ \sum_{\bv{\omega}\in\Omega| \omega_i=v_i} \sum_{\bv{\alpha} \in \script{A}} \pi[\bv{\omega}]Pr[\bv{\alpha}|\bv{\omega}]$.  
Then constraints \eqref{eq:NE-better2} hold by construction.  Now let $r_i$ be the index for the pure strategy of player $i$ that 
selects actions according to the function $b_i(v_i)$.  Then: 
\begin{align} 
&\sum_{\bv{\omega} \in \Omega | \omega_i = v_i}\sum_{\bv{\alpha}\in\script{A}}\pi[\bv{\omega}]Pr[\bv{\alpha}|\bv{\omega}]\theta_i(v_i) \nonumber \\
&= \sum_{\bv{\omega}\in\Omega|\omega_i=v_i}\sum_{\bv{\alpha}\in\script{A}} \pi[\bv{\omega}]Pr[\bv{\alpha}|\bv{\omega}]\hat{u}_i((b_i^{(r_i)}(v_i), \bv{\alpha}_{\overline{i}}), \bv{\omega})
\end{align} 
Summing the above over all $v_i \in \Omega_i$ gives: 
\begin{align} 
&\sum_{\bv{\omega} \in \Omega}\sum_{\bv{\alpha}\in\script{A}}\pi[\bv{\omega}]Pr[\bv{\alpha}|\bv{\omega}]\theta_i(\omega_i) \nonumber \\
&= \sum_{\bv{\omega}\in\Omega}\sum_{\bv{\alpha}\in\script{A}} \pi[\bv{\omega}]Pr[\bv{\alpha}|\bv{\omega}]\hat{u}_i((b_i^{(r_i)}(\omega_i), \bv{\alpha}_{\overline{i}}), \bv{\omega}) \nonumber \\
&= \sum_{\bv{s}\in\script{S}} Pr[\bv{s}]h_i(r_i, \bv{s}_{\overline{i}})  \label{eq:dood1} \\
&\leq \sum_{\bv{s}\in\script{S}} Pr[\bv{s}]h_i(\bv{s}) \label{eq:dood2} \\
&= \sum_{\bv{\omega}\in\Omega}\sum_{\bv{\alpha}\in\script{A}} \pi[\bv{\omega}]Pr[\bv{\alpha}|\bv{\omega}]\hat{u}_i(\bv{\alpha}, \bv{\omega}) \label{eq:dood3} 
\end{align} 
where \eqref{eq:dood1} follows from \eqref{eq:useful-id2}, \eqref{eq:dood2} follows from \eqref{eq:NE-virtual}, 
and \eqref{eq:dood3} follows from \eqref{eq:useful-id1}.  Thus, the constraints \eqref{eq:NE-better1} hold. 
\end{proof} 

The proof of Lemma \ref{lem:equiv-virtual}b is similar to that of Lemma \ref{lem:equiv-virtual}a and is omitted for brevity.

\begin{proof} (Lemma \ref{lem:equiv-virtual}c)
Suppose $Pr[\bv{s}]$ is a NE of the virtual static game.  Then it has the product form \eqref{eq:product-form-virtual}. 
Lemma \ref{lem:product-form} ensures that the corresponding 
$Pr[\bv{\alpha}|\bv{\omega}]$ has the product form \eqref{eq:product-form-conditional}.  Further, because $Pr[\bv{s}]$ satisfies
\eqref{eq:NE-virtual}, 
Lemma \ref{lem:equiv-virtual}a ensures that $Pr[\bv{\alpha}|\bv{\omega}]$  satisfies \eqref{eq:NE-better1}-\eqref{eq:NE-better2}. 
Hence, $Pr[\bv{\alpha}|\bv{\omega}]$ is a NE for the stochastic game. 
\end{proof} 

\section{Appendix B --- Proof that $\script{E}_{NE}^{stoc} \subseteq \script{E}_{CE}^{stoc}$}

Suppose that $Pr[\bv{\alpha}|\bv{\omega}]$ is a NE for the stochastic game, so that it has
the product form \eqref{eq:product-form-conditional} and satisfies the constraints \eqref{eq:NE-better1}-\eqref{eq:NE-better2}.  Suppose $Pr[\bv{\alpha}|\bv{\omega}]$ is \emph{not} a CE for the stochastic game, so that 
constraints \eqref{eq:CE-better1}-\eqref{eq:CE-better2} are \emph{not} satisfied.  The goal is to reach a contradiction. 

Since the constraints \eqref{eq:CE-better1}-\eqref{eq:CE-better2} are not satisfied, by Lemma \ref{lem:equiv-all3} 
it follows that there exists a player $i \in \{1, \ldots, N\}$ and a  randomized function $X_i(\omega_i, \alpha_i)$ for which: 
\begin{align}
&\expect{u_i(t)} <  \nonumber \\
&\sum_{\bv{\omega}\in\Omega}\sum_{\bv{\alpha}\in \script{A}} \pi[\bv{\omega}] Pr[\bv{\alpha}|\bv{\omega}]\expect{\hat{u}_i((X_i(\omega_i, \alpha_i), \bv{\alpha}_{\overline{i}}), \bv{\omega})} \label{eq:fooFOO} 
\end{align}
Because the product form property \eqref{eq:product-form-conditional} holds, one has: 
\begin{eqnarray}
Pr[\bv{\alpha}|\bv{\omega}] = Pr[\alpha_i|\omega_i]\prod_{j\neq i} Pr[\alpha_j|\omega_j] \label{eq:foopers} 
\end{eqnarray}
Now consider the following alternative strategy for player $i$, based only on knowledge of its observed $\omega_i$ (without knowledge of $\alpha_i$):  Define the random action $\tilde{X}_i(\omega_i)$ that observes $\omega_i$, independently generates 
a random element $\tilde{\alpha}_i \in \script{A}_i$ according to the conditional distribution $Pr[\alpha_i|\omega_i]$, and then defines $\tilde{X}_i(\omega_i) = X_i(\omega_i, \tilde{\alpha}_i)$. This policy uses the conditional distribution associated with the actual $\alpha_i$ value, but does not require knowledge of this actual value. 
Define $val_i$ as the expected utility of player $i$ under this alternative randomized strategy $\tilde{X}_i(\omega_i)$: 
\begin{align*} 
&val_i \defequiv \sum_{\bv{\omega} \in \Omega} \pi[\bv{\omega}]\sum_{\bv{\alpha}_{\overline{i}} \in \script{A}_{\overline{i}}} \sum_{\tilde{\alpha}_i \in \script{A}_i} \prod_{j\neq i} Pr[\alpha_j|\omega_j] Pr[\tilde{\alpha}_i|\omega_i] \times \\
&  \hspace{+.5in} \expect{\hat{u}_i((X_i(\omega_i, \tilde{\alpha}_i), \bv{\alpha}_{\overline{i}}), \bv{\omega})} 
\end{align*}
By \eqref{eq:foopers}, $val_i$ is the same as the right-hand-side of \eqref{eq:fooFOO}, and so $\expect{u_i(t)}< val_i$. 
On the other hand, since $Pr[\bv{\alpha}|\bv{\omega}]$ is a NE for the stochastic game, Lemma \ref{lem:equiv-all1} 
ensures that no such alternative randomized strategy $\tilde{X}_i(\omega_i)$ can improve the average utility of player $i$, so that $\expect{u_i(t)} \geq val_i$.   This is the desired contradiction.

\bibliographystyle{unsrt}
\bibliography{../../latex-mit/bibliography/refs}
\end{document}